\documentclass[acmsmall]{acmart}

\usepackage{lineno}   % dzl
\usepackage{booktabs} % For formal tables
\usepackage{url}  %gws 20170930

\usepackage{multirow}  %gws 20170930
\usepackage{color}  %gws 20170930
\newtheorem{strategy}{Strategy}    % gws
\newtheorem{upper bound}{Upper bound}

\usepackage{diagbox}
\usepackage{amsmath}
\usepackage{arydshln} % gws, 20180202
\usepackage{color} % gws, 20180120

\usepackage[linesnumbered,ruled]{algorithm2e} % For algorithms

\SetAlFnt{\small}
\SetAlCapFnt{\small}
\SetAlCapNameFnt{\small}
\SetAlCapHSkip{0pt}
\IncMargin{-\parindent}

\usepackage{xspace}
\usepackage{multirow}
\usepackage{subfigure}

\usepackage{enumitem}

\setcopyright{acmcopyright}
\acmJournal{JACM}

\begin{document}

\title{Towards Sequence Utility Maximization under Utility Occupancy Measure}

\author{Gengsen Huang}
\affiliation{ %College of Cyber Security, 
	\institution{Jinan University}
	\city{Guangzhou}
	\state{Guangdong}
	\country{China}
}
\email{hgengsen@gmail.com}

\author{Wensheng Gan}
\authornote{This is the corresponding author}
\affiliation{
	\institution{Jinan University}
	\city{Guangzhou}
	\state{Guangdong}
	\country{China}
}
\email{wsgan001@gmail.com}

\author{Philip S. Yu}
\affiliation{%
	\institution{University of Illinois at Chicago}
	\city{Chicago}
	\country{USA}
}
\email{psyu@uic.edu}

%%
%% By default, the full list of authors will be used in the page
%% headers. Often, this list is too long, and will overlap
%% other information printed in the page headers. This command allows
%% the author to define a more concise list
%% of authors' names for this purpose.
%\renewcommand{\shortauthors}{Huang and Gan, et al.}

%%
%% The abstract is a short summary of the work to be presented in the
%% article.
\begin{abstract}
	The discovery of utility-driven patterns is a useful and difficult research topic. It can extract significant and interesting information from specific and varied databases, increasing the value of the services provided. In practice, the measure of utility is often used to demonstrate the importance, profit, or risk of an object or a pattern. In the database, although utility is a flexible criterion for each pattern, it is a more absolute criterion due to the neglect of utility sharing. This leads to the derived patterns only exploring partial and local knowledge from a database. Utility occupancy is a recently proposed model that considers the problem of mining with high utility but low occupancy. However, existing studies are concentrated on itemsets that do not reveal the temporal relationship of object occurrences. Therefore, this paper towards sequence utility maximization. We first define utility occupancy on sequence data and raise the problem of High Utility-Occupancy Sequential Pattern Mining (HUOSPM). Three dimensions, including frequency, utility, and occupancy, are comprehensively evaluated in HUOSPM. An algorithm called Sequence Utility Maximization with Utility occupancy measure (SUMU) is proposed. Furthermore, two data structures for storing related information about a pattern, Utility-Occupancy-List-Chain (UOL-Chain) and Utility-Occupancy-Table (UO-Table) with six associated upper bounds, are designed to improve efficiency. Empirical experiments are carried out to evaluate the novel algorithm's efficiency and effectiveness. The influence of different upper bounds and pruning strategies is analyzed and discussed. The comprehensive results suggest that the work of our algorithm is intelligent and effective.
\end{abstract}

%%
%% Keywords. The author(s) should pick words that accurately describe
%% the work being presented. Separate the keywords with commas.
\keywords{Pattern discovery, sequential pattern, utility mining, utility occupancy}

\authorsaddresses{Authors' addresses: G. Huang and W. Gan, College of Cyber Security, Jinan University, Guangzhou, China; email: hgengsen@gmail.com and wsgan001@gmail.com. Philip S. Yu, University of Illinois at Chicago, Chicago, USA; email: psyu@uic.edu
}

%%
%% The code below is generated by the tool at http://dl.acm.org/ccs.cfm.
%% Please copy and paste the code instead of the example below.
%%
\begin{CCSXML}
<ccs2012>
   <concept>
       <concept_id>10002951.10003317</concept_id>
       <concept_desc>Information systems~Data mining</concept_desc>
       <concept_significance>500</concept_significance>
       </concept>
   <concept>
       <concept_id>10010147.10010257</concept_id>
       <concept_desc>Computing methodologies~Machine learning</concept_desc>
       <concept_significance>300</concept_significance>
       </concept>
 </ccs2012>
\end{CCSXML}

\ccsdesc[500]{Information Systems~Data mining}

\ccsdesc[300]{Applied computing~Business intelligence} 

% \ccsdesc[300]{Computing methodologies~Machine learning}

%%
%% This command processes the author and affiliation and title
%% information and builds the first part of the formatted document.
\maketitle

\section{Introduction}

Sequential pattern mining (SPM) \cite{agrawal1995mining,fournier2017survey} is a fundamental and important technique to explore useful knowledge from a wide variety and huge amount of data. The patterns derived by SPM can reveal the temporal relationship of objects or events, solving the dilemmas and limitations of frequent itemset mining (FIM) \cite{agrawal1994fast, han2004mining, geng2006interestingness}. Generally speaking, FIM aims to mine collections containing multiple objects or events, while the goal of SPM is to discover sequences comprising multiple collections. Consideration of the chronological order renders SPM more challenging but also more rewarding. SPM is a discovery process based on user-predefined minimum support (\textit{minsup}) and the frequency of candidate patterns. The research on the frequency evaluation metric has been widely studied and has also been extensively extended to meet the various needs of users, such as closed sequential pattern \cite{wang2007frequent}, top-$k$ sequential pattern \cite{petitjean2016skopus}, and constraint-based sequential pattern \cite{van2018mining,wu2021ntp}. Besides, sequence prediction is an interesting derivative direction that attempts to predict what object or event will happen next.

In recent years, as various applications and scenario cases have been studied in detail, the importance of frequency-oriented SPM has no longer been quite conspicuous. In practical services, various objects are often influenced by disparate implicit factors. Therefore, traditional SPM algorithms are not suitable for complex mining tasks. The utility model adequately takes into account the weight, profit, or risk that an object may actually hold and was first introduced into mining as a framework, namely high-utility pattern mining (HUPM) \cite{gan2019survey, fournier2019survey}. In general, the utility of each object is set according to the preferences selected by the user, or its decision attributes. Moreover, frequency-oriented mining can be seen as a special task of utility-oriented mining, i.e., when the utilities of all objects are set to a constant one. Utility-oriented mining is an emerging research issue that is also quickly applied to transaction data \cite{tseng2010up, tseng2012efficient, liu2012mining, zida2017efim} and sequence data \cite{yin2012uspan, wang2016efficiently, gan2020proum, gan2020fast}. High-utility sequential pattern mining (HUSPM) is highly regarded and also applied to many applications \cite{shie2011mining, zihayat2017mining} due to its outstanding characteristics. Compared to traditional SPM, HUSPM is much more difficult. In a mined database, the utility of each object in the sequence is not uniform. Obviously, this configuration is more relevant to real-world applications. For instance, something that occurs today with high utility will not occur tomorrow. Furthermore, HUSPM has many interesting contents when combined with other theories of knowledge or superior industrial technologies, such as fuzzy theory \cite{gan2021explainable}, average utility model \cite{truong2020ehausm}, the nettree structure \cite{wu2021hanp}, and computational framework processing \cite{srivastava2020large}.

The main problem with HUSPM is that it can only explore local knowledge, ignoring the influence of the global database on the patterns. For example, in a sequence database, it is easy to calculate the utility of a sequential pattern, and the pattern is deemed a quality pattern if its utility is high enough. Thus, irrelevant objects in the sequence records of the database obviously have little effect on the calculation. It means that even if a database contains a large number of irrelevant objects for this pattern, they will not have any influence on the information carried by this pattern. Recently, a measure known as occupancy \cite{zhang2015occupancy} has been proposed to view global database knowledge. For example, for a pattern of length 3, if the lengths of the transaction records containing it are 3, 5, and 6, the occupancy of this pattern is then equal to $\frac{1}{3}$ $\times$ ($\frac{3}{3}$ + $\frac{3}{5}$ + $\frac{3}{6}$) = 0.7. Those records that contain more irrelevant items make the occupancy of the pattern lower. Although it measures the completeness of a pattern in the database, it is a simple complement of the frequency or support. Interesting information such as weight, profit, and risk is still ignored. It suggests that it is wise and helpful to mine patterns that share a high occupancy on utilities. The concept of utility occupancy \cite{shen2016ocean, gan2019huopm} is defined subsequently. Utility occupancy is introduced into transactional quantitative databases, making such technologies capable of mining patterns that are frequent and highly utility-occupied. However, the existing studies have not been adapted to sequence databases, only addressing itemset mining.

In this paper, we toward sequence utility maximization under the utility occupancy measure. We propose a generic framework to discover high utility-occupancy sequential patterns (HUOSPs) in sequence data, thus addressing the lack of existing research. The concept of HUOSP is different from high utility-occupancy itemset (HOUI) and high-utility sequential pattern (HUSP). Since the number of occurrences and utility of each object in the sequence record are different, it causes HUSPM and its extensions to be more challenging, especially in utility calculation or other calculations related to utility. Utility occupancy on sequence data also faces such problems, making mining tasks more complicated to tackle. The major contributions of this article are summarized as follows:

\begin{itemize}
	\item The concept of utility occupancy is applied toward sequence utility maximization, and we present the related concepts and definitions. A novel problem for mining HUOSPs is formalized, which means taking account of the utility occupancy measure. According to our survey, high utility-occupancy mining is an emerging and promising topic in pattern discovery. No prior work has been successful in exploring the sequence data.
	
	\item Two compact data structures, Utility-Occupancy-List-Chain (UOL-Chain) and Utility-Occupancy-Table (UO-Table), are designed for storing the essential information about candidate patterns in the mining process. Six upper bounds for support and utility occupancy are also proposed to make sure that mining results are correct and complete. 
	
	\item A novel algorithm called Sequence Utility Maximization with Utility Occupancy Measure (SUMU) is proposed to discover interesting and interpretable HUOSPs that have high utility-occupancy and frequency. Furthermore, we develop some pruning strategies to boost the efficiency of the proposed algorithm based on the designed upper bounds and data structures.
	
	\item Extensive experiments are conducted to fully analyze the influence of the novel algorithm under both support and utility-occupancy measures. Moreover, four variants of SUMU using different strategies are compared in terms of several aspects.
\end{itemize}

The paper is briefly organized as follows: The previous work related to this paper is reviewed and summarized in Section \ref{sec:relatedwork}. Then, in Section \ref{sec:preliminaries}, the fundamental preliminaries are given, and the problem is formalized. The designed data structures and upper bounds are described with the proposed HOUSPM algorithm in Section \ref{sec:algorithm}. We evaluate the effectiveness and efficiency of the SUMU algorithm and perform experimental analysis on different sequence datasets in Section \ref{sec:experiments}. Finally, in Section \ref{sec:conclusion}, we draw the conclusions and discuss some potential future work.

\section{Related Work}  \label{sec:relatedwork}

\subsection{High-Utility Sequential Pattern Mining}
\label{relWork:HUSPM} 

Considering the order of items in sequence data, sequential pattern mining \cite{agrawal1995mining} was proposed. SPM has been widely applied and served for many applications, including business analysis \cite{fournier2017survey, gan2019surveyp} and medical analysis \cite{gao2020toward}. Later, as the technology of SPM became more sophisticated and computer equipment was upgraded, users focused more on interesting patterns. Constraint-based sequential pattern mining \cite{pei2002mining, garofalakis2002mining} produces more concise and reasonable patterns by imposing a series of constraints. However, the main problem with SPM is that it treats all items that appear in the database as equally important. This means that the mining tasks based on the support measure seem inadequate. By introducing utility evaluation into the sequence data, an important issue was developed, namely high-utility sequential pattern mining (HUSPM). HUSPM can mine those high-yield patterns from the database and the utilities assigned by users for different items whose utilities are determined by their profits or risks. UtilityLevel and UtilitySpan were proposed by Ahmed \textit{et al.} \cite{ahmed2010novel} for mining HUSPs. UtilityLevel is a level-wise-based approach and UtilitySpan utilizes the idea of pattern-growth. In addition, two tree-based approaches \cite{shie2011mining} were proposed to obtain interesting patterns in the mobile commerce environment. These two methods use depth-first and breadth-first strategies, respectively. And then, Yin \textit{et al.} \cite{yin2012uspan} provided a generic framework for HUSPM and designed a fast algorithm called USpan. Related concatenation mechanisms and pruning strategies were designed to quickly calculate the utility of a node and its children in a tree data structure (called a lexicographic q-sequence tree). After two database scans, the information about the sequence is stored in a utility matrix. Besides, an upper bound on utility calculation was designed to quickly discover HUSPs. However, the value of the upper bound is overestimated for some candidate patterns, resulting in the omission of some true HUSPs \cite{gan2020proum}.

A projection-based approach called PHUS was proposed by Lan \textit{et al.} \cite{lan2014applying}. It mainly used a maximum utility measure and an efficient indexing structure to expedite the entire mining process, and also proposed the sequence-utility upper bound (SUUB). Then, HuspExt \cite{alkan2015crom} was proposed, which is based on the upper bound cumulated rest of match (CRoM) to eliminate candidate items and patterns. Subsequently, an algorithm with the utility-chain data structure called HUS-Span was proposed by Wang \textit{et al.} \cite{wang2016efficiently}. There are two tight upper bounds introduced in HUS-Span, and thus HUS-Span can quickly identify HUSPs with the help of pruning strategies. Inspired by the idea of projection, the ProUM algorithm \cite{gan2020proum} introduced the upper bound sequence extension utility (SEU) to eliminate unpromising sequences. The projection mechanism used in the designed utility-array and ProUM can work well in the mining process. Then, an efficient HUSP-ULL algorithm \cite{gan2020fast} was proposed. The designed UL-list can be used to efficiently discover the entire set of HUSPs. In addition, two pruning strategies (LAR and IIP) were introduced to avoid useless pattern extensions. Aside from these HUSPM algorithms, the TKHUS-Span algorithm \cite{wang2016efficiently} with three search methods was developed to identify the top-$k$ HUSPs in a sequence database. IncUSP-Miner+ \cite{wang2018incremental} aims to deal with mining tasks in a dynamic environment.

\subsection{High Utility-Occupancy Pattern Mining}
\label{relWork:HUOPM} 

As research on high-utility pattern mining (HUPM) has continued to intensify, researchers have realized the deficiencies of the measure of utility. HUPM often encounters many dilemmas, such as the rare item problem \cite{weiss2004mining, gan2021beyond}, lack of correlation \cite{fournier2016mining, gan2019correlated}, and neglect of an intrinsic relationship \cite{wu2010re, kim2011efficient}. Hence, there are many measures that have been proposed. Occupancy is a flexible and interesting measure that was first introduced by Zhang \textit{et al.} \cite{zhang2015occupancy}. With the occupancy measure, the completeness of a pattern can be evaluated, and the mined patterns are deemed dominant and frequent. Compared to the support measure, the important downward closure property is not valid for occupancy. This means that the estimated values of the occupancy calculation need to be explored. The DOFRA algorithm \cite{zhang2015occupancy} discussed different data situations and proposed relevant upper bounds. However, it also suffers from some problems with frequent pattern mining.

For utility mining, Shen \textit{et al.} \cite{shen2016ocean} defined the measure of utility occupancy for the first time, and proposed the OCEAN algorithm that is based on the utility-list. It also derived an upper bound to evaluate the likely contribution of a pattern in the HUPM tasks. However, due to the use of inconsistent sorting orders, OCEAN obtains an incomplete result. Besides, the efficiency of OCEAN is not good enough to make good use of the support property and utility occupancy property. \textit{Gan} \textit{et al.} \cite{gan2019huopm} introduced some tight upper bounds based on the properties of support and utility occupancy. An efficient algorithm called HUOPM was also proposed. Two list-based data structures and a frequency-utility tree were designed in HUOPM to store important information about patterns. Chen \textit{et al.} \cite{chen2021discovering} then explored HUOPM on the uncertain data. Three useful factors, including utility contribution, frequency, and probability, were considered, and the UHUOPM algorithm was proposed to obtain all potential HUOPs. To obtain more flexible HUOPs, the HUOPM$^+$ algorithm\cite{chen2021flexible} was devised, which takes into account the minimum and maximum length constraints.

\section{Preliminaries and Problem Statement}
\label{sec:preliminaries}

In this section, we first introduce and define the basic notations and concepts related to utility occupancy mining on sequence data. The problem of high utility-occupancy sequential pattern mining is then formulated.

\subsection{Notations and Concepts}

Given a finite set $I$ = \{$i_{1}$, $i_{2}$, $\cdots$, $i_{m}$\} containing $m$ distinct items, a quantitative itemset $c$  is a non-empty set and can be defined as $c$ = [($i_1$, $q_1$)($i_2$, $q_2$)$\cdots$($i_n$, $q_n$)], where $q_j$ is the quality value for $i_j$. Each item and its associated quality (internal utility) together comprise the elements of the quantitative itemset $c$. The items in the quantitative itemset $c$ is a subset of $I$. An itemset $w$ is a non-empty set with no quality information for $c$, which is called that $w$ matches $c$, and is denoted as $w$ $\sim$ $c$. To simplify the description of some definitions in this paper, we assume that all items in a quantitative itemset are sorted alphabetically. A quantitative sequence is denoted as $s$ and defined as $s$ = $<$$c_1$, $c_1$, $\cdots$, $c_l$$>$. $s$ is an ordered list containing one or more quantitative itemsets, and the order in which the quantitative itemsets appear can represent the chronological relationship of realistic applications. $v$ = $<$$w_1$, $w_1$, $\cdots$, $w_l$$>$ is used as $s$ without quantity information that is called that $v$ matches $s$ and is denoted as $v$ $\sim$ $s$. For the sake of illustration, quantitative itemset and quantitative sequence can also be termed as $q$-itemset and $q$-sequence. Regarding a quantitative sequence database $\mathcal{D}$, it is a collection of triples $<$\textit{SID}, \textit{qs}, \textit{SU}$>$, where \textit{qs} is a $q$-sequence, \textit{SID} is the unique identifier of \textit{qs}, and \textit{SU} is the total utility of \textit{qs}. Furthermore, each item $i$ such that $i$ $\in$ $\mathcal{D}$ has its own profit value (called external utility), and can be denoted as $p$($i$).

\begin{table}[h]
	\centering
	\caption{Quantitative sequence database}
	\label{table1}
	\begin{tabular}{|c|c|c|}  
		\hline 
		\textbf{SID} & \textbf{Quantitative sequence} & \textbf{SU} \\
		\hline  
		\(s_{1}\) & $<$[(\textit{b}, 2)(\textit{d}, 1)], [(\textit{g}, 1)], [(\textit{f}, 1)]$>$ & 11 \\ 
		\hline
		\(s_{2}\) & $<$[(\textit{d}, 1)], [(\textit{g}, 1)]$>$  & 2 \\  
		\hline  
		\(s_{3}\) & $<$[(\textit{a}, 1)(\textit{b}, 1)], [(\textit{c}, 1)], [(\textit{c}, 2)], [(\textit{d}, 1)]$>$ & 12 \\
		\hline  
		\(s_{4}\) & $<$[(\textit{a}, 2)(\textit{b}, 1)], [(\textit{c}, 1)], [(\textit{e}, 1)]$>$ & 13 \\
		\hline
		\(s_{5}\) & $<$[(\textit{d}, 3)], [(\textit{b}, 1)], [(\textit{a}, 1)], [(\textit{c}, 1)], [(\textit{e}, 1)]$>$ & 13 \\
		\hline
	\end{tabular}
\end{table}

\begin{table}[h]
	\caption{External utility table}
	\label{table2}
	\centering
	\begin{tabular}{|c|c|c|c|c|c|c|c|}
		\hline
		\textbf{Item}	    & \textit{a}	& \textit{b}	& \textit{c}	& \textit{d}	& \textit{e}	& \textit{f}  & \textit{g} \\ \hline 
		\textbf{\textit{Unit utility}}	& 3 & 2 & 2 & 1 & 3 & 5 & 1 \\ \hline
	\end{tabular}
\end{table}

The example $q$-sequence database and external utility table that will be used in the following are shown in Tables \ref{table1} and \ref{table2}. We can see that this database has five $q$-sequences and seven different items. [($b$, 2)($d$, 1)] is the first $q$-itemset in $q$-sequence $s_1$, containing two items, $b$ and $d$. According to Table \ref{table2}, the external utility of items $b$ and $d$ are 2 and 1, respectively. In addition, $<$[$b$$d$]$>$ matches $<$[($b$, 2) ($d$, 1)]$>$.

\begin{definition}
	\rm For an item $i$ in a $q$-itemset $c$, its utility can be denoted as $u$($i$, $c$) and is defined as $u$($i$, $c$) = $q$($i$, $c$) $\times$ $p$($i$, $c$) where $q$($i$, $c$) is the internal utility of $i$ in $c$ and $p$($i$, $c$) is the external utility of $i$. We use $u$($c$) to denote the sum of utilities of all items in $c$, and it can be defined as $u$($c$) = $\sum\limits_{i \in c}u(i, c)$. As for a $q$-sequence $s$, its utility can be denoted as $u$($s$) and is defined as $u$($s$) = $\sum\limits_{c \in s}u(c)$. Moreover, given a $q$-sequence database $\mathcal{D}$, its utility can be denoted as $u$($\mathcal{D}$) and is defined as $u$($\mathcal{D}$) = $\sum\limits_{s \in \mathcal{D}}u(s)$.
\end{definition}

For example, the utility of the item $b$ is equal to 4, because $u$($b$, $s_1$) = 2 $\times$ 2 = 4; the utilities of three $q$-itemsets in $s_1$ are 5, 1, and 5, respectively. Thus, the \textit{SU} of $s_1$ can be calculated as $u$($s_1$) = 5 + 1 + 5 = 11; the total utility of this example database $\mathcal{D}$ is calculated as $u$($\mathcal{D}$) = $\sum_{s_i \in \mathcal{D}}$ $u$($s_i$) = 11 + 2 + 12 + 13 + 13 = 51.

\begin{definition}
	\rm Given two itemsets $w$ and $w^\prime$, if all the items of $w$ appear in $w^\prime$, we say that $w^\prime$ contains $w$, and is denoted as $w$ $\subseteq$ $w^\prime$. Similarly, for two $q$-itemset $c$ and $c^\prime$, if all the items of $c$ appear in $c^\prime$ and have the same quality, we say that $c^\prime$ contains $c$, which is denoted as $c$ $\subseteq$ $c^\prime$.
\end{definition}

For instance, the itemset [$c$$d$$e$] contains the itemset [$c$$e$]. And the $q$-itemset [($c$, 4)($e$, 2)] is contained in [($c$, 4)($d$, 3)($e$, 2)], but not in [($c$, 3)($e$, 3)]. Because the quality of $c$ in these two $q$-itemsets [($c$, 3)($e$, 3)] and [($c$, 4)($d$, 3)($e$, 2)] is different.

\begin{definition}
	\rm Given two sequences $v$ = $<$$w_1$, $w_2$, $\cdots$, $w_l$$>$ and $v^\prime$ = $<$$w^\prime_1$, $w^\prime_2$, $\cdots$, $w^\prime_l$$>$, if there exists an integer list (1 $\le$ $k_1$ $\le$ $k_2$ $\le$ $\cdots$ $l$) satisfies that $w_j$ $\subseteq$ $w^\prime_{k_j}$, 1 $\le$ $j$ $\le$ $l$, we say that $v^\prime$ contains $v$, and is denoted as $v$ $\subseteq$ $v^\prime$. For two $q$-sequences $s$ = $<$$c_1$, $c_1$, $\cdots$, $c_l$$>$ and $s^\prime$ = $<$$c^\prime_1$, $c^\prime_1$, $\cdots$, $c^\prime_l$$>$, if these two $q$-sequences need to satisfy the containment relationship, then there exists an integer list (1 $\le$ $k_1$ $\le$ $k_2$ $\le$ $\cdots$ $l$) satisfies that $c_j$ $\subseteq$ $c^\prime_{k_j}$, 1 $\le$ $j$ $\le$ $l$, which is denoted as $s$ $\subseteq$ $s^\prime$. In this paper, if a sequence $t$ matches a $q$-sequence $s_k$ and also satisfies $s_k$ $\subseteq$ $s$, then it can also be denoted as $t$ $\subseteq$ $s$ instead of $t$ $\sim$ $s_k$ $\land$ $s_k$ $\subseteq$ $s$.
\end{definition}

For example, the $q$-sequence $s_1$ contains $<$[($b$, 2)($d$, 1)]$>$ and $<$[($g$, 1)], [($f$, 1)]$>$, while $<$[($b$, 2)($d$, 2)]$>$ and $<$[($g$, 1)($f$, 1)]$>$ are not contained in $s_1$. 

\begin{definition}
	\rm For a sequences $t$, it has multiple matches in a $q$-sequence $s$. We use $u$($t$, $s$) to denote the actual utility of $s$ and it is defined as $u$($t$, $s$) = \textit{max}\{$u$($s^\prime$) $\vert$ $t$ $\sim$ $s^\prime$ $\land$ $s^\prime$ $\subseteq$ $s$\}. Additionally, the utility of $t$ in the $q$-sequence database $\mathcal{D}$ can be denoted as $u$($t$) and is defined as $u$($t$) = \{$\sum\limits_{s \in \mathcal{D}}u(t, s) \vert t \subseteq s$\}. In addition, its support can be denoted as \textit{sup}($t$) and is defined as \textit{sup}($t$) = $\vert$ $t$ $\subseteq$ $s$ $\land$ $s$ $\in$ $\mathcal{D}$ $\vert$, that is, the number of $q$-sequences of $\mathcal{D}$ matching $t$.
\end{definition}

For example, the sequence $t$ = $<$[$a$$b$], [$c$]$>$ has two matches in the $q$-sequence $s_3$, and so its utility can be calculated as $u$($<$[$a$$b$], [$c$]$>$) = \textit{max}\{$u$($<$[($a$, 1)($b$, 1)], [($c$, 1)]$>$), $u$($<$[($a$, 1)($b$, 1)], [($c$, 2)]$>$)\} = \textit{max}\{7, 9\} = 9. And $t$ has a support of 2 because $s_3$ and $s_4$ both have instances where $t$ matches.

In this paper, the concept of utility occupancy \cite{gan2019huopm} is incorporated into sequence data. Utility occupancy is a flexible measure that can be used to identify patterns with a higher contribution in sequences. Since there is no previous work on this topic, we are the first to define the relevant concepts.

\begin{definition}
	\rm In a $q$-sequence $s$, the utility occupancy of a sequence $t$, denoted as \textit{uo}($t$, $s$), is defined as $\textit{uo}(t, s)$ = $\frac{u(t, s)}{u(s)}$. Note that $t$ may have more than one match at $s$. Then the utility occupancy of $t$ at position $p$ in $s$ can be denoted as \textit{uo}($t$, $s$, $p$) and is defined as follows.
	$$
	\begin{aligned}
	\textit{uo}(t, s, p) = \frac{\textit{max}\{u(t, s^\prime) \vert t \sim s^\prime \land s^\prime \subseteq <s_{1}, \cdots, s_{p}>\}}{u(s)}.
	\end{aligned}
	$$	
	The total utility occupancy of $t$ in a $q$-sequence database $\mathcal{D}$, denoted as \textit{uo}($t$), is defined as follows.	
	$$
	\begin{aligned}
	\textit{uo}(t) = \frac{\sum\limits_{t \subseteq s \land s \in \mathcal{D}}\textit{uo}(t, s)}{\textit{sup}(t)}.
	\end{aligned}
	$$
\end{definition}

For example, the utility occupancy of the sequence $<$[$a$], [$c$]$>$ in $s_3$, $s_4$, and $s_5$ are \textit{uo}($<$[$a$], [$c$]$>$, $s_3$) = \textit{max}(\{5, 7\}) / 12 = 0.583, \textit{uo}($<$[$a$], [$c$]$>$, $s_4$) =  8 /13 = 0.615, and \textit{uo}($<$[$a$], [$c$]$>$, $s_5$) = 5 / 13 = 0.385, respectively. Thus, the total utility occupancy of the sequence $<$[$a$], [$c$]$>$ in the entire $\mathcal{D}$ is equal to \textit{uo}($<$[$a$], [$c$]$>$) = (0.583 + 0.615 + 0.385) / 3 = 0.528.

\begin{definition}
	\rm In a $q$-sequence $s$ with $l$ $q$-itemsets, the remaining utility occupancy of a sequence $t$ at position $p$ can be denoted as \textit{ruo}($t$, $s$, $p$), and is defined as follows.
	$$
	\begin{aligned}
	\textit{ruo}(t, s, p) = \frac{u(<s_{p+1}, \cdots, s_{l}>, s)}{u(s)}.
	\end{aligned}
	$$
\end{definition}

\begin{definition}
	\rm Considering two thresholds, including a minimum support threshold \textit{minsup} (0 $\textless$ \textit{minsup} $\le$ 1) and a minimum utility occupancy threshold \textit{minuo} (0 $\textless$ \textit{minuo} $\le$ 1), a sequential pattern $t$ with high support and high utility occupancy in a $q$-sequence database $\mathcal{D}$ is called a HUOSP. Here, it satisfies \textit{sup}($t$) $\ge$ \textit{minsup} and \textit{uo}($t$) $\ge$ \textit{minuo}.
\end{definition}

For example, the remaining utility occupancy of the sequence $<$[$a$], [$c$]$>$ in $s_3$ at position 2 is equal to \textit{ruo}($<$[$a$], [$c$]$>$, $s_3$, 2) = (4 + 1) / 12 = 0.417. And the remaining utility occupancy of the sequence $<$[$a$]$>$ in $s_4$ at position 1 is equal to \textit{ruo}($<$[$a$]$>$, $s_4$, 1) = (2 + 2 + 3) / 13 = 0.538. Under the setting of \textit{minsup} to 2 and \textit{minuo} to 0.4, all found HUOSPs are shown in Table \ref{table_huosp}.

\begin{table}[h]
	\centering
	\caption{All found HUOSPs in the example database $\mathcal{D}$}
	\label{table_huosp}
	\begin{tabular}{|c|c|c|c|}  
		\hline 
		\textbf{ID} & \textbf{HUOSP} & \textbf{Support} & \textbf{Utility occupancy}\\
		\hline  
		\(p_{1}\) & $<$[$a$$b$]$>$ & 2 & 0.516 \\ 
		\hline
		\(p_{2}\) & $<$[$a$$b$], [$c$]$>$ & 2 & 0.76 \\ 
		\hline
		\(p_{3}\) & $<$[$a$], [$c$]$>$ & 3 & 0.528 \\ 
		\hline
		\(p_{4}\) & $<$[$a$], [$c$], [$e$]$>$ & 2 & 0.731 \\ 
		\hline
		\(p_{5}\) & $<$[$a$], [$e$]$>$ & 2 & 0.577 \\ 
		\hline
		\(p_{6}\) & $<$[$b$], [$c$], [$e$]$>$ & 2 & 0.538 \\ 
		\hline
		\(p_{7}\) & $<$[$d$], [$g$]$>$ & 2 & 0.59 \\ 
		\hline
	\end{tabular}
\end{table}

\subsection{Problem Statement}

After the above definitions and concepts are given, we formulate the problem of mining HUOSPs as follows. Given a quantitative sequence database $\mathcal{D}$, a utility table with external utilities for each item, and two thresholds \textit{minsup} (0 $\textless$ \textit{minsup} $\le$ 1) and \textit{minuo} (0 $\textless$ \textit{minuo} $\le$ 1), the goal of high utility-occupancy sequential pattern mining is to discover all HUOSPs that frequency and utility occupancy are greater than \textit{minsup} and \textit{minuo}, respectively.

\section{Proposed SUMU Algorithm}
\label{sec:algorithm}

In this section, we present an algorithm called Sequence Utility Maximization with Utility Occupancy Measure (SUMU) to discover interesting HUOSPs. Some measures of support and utility occupancy are used to reduce unnecessary operations and prune the pattern search space. Besides, we design two data structures that can record some essential information about a candidate pattern. They are called Utility-Occupancy-List-Chain (UOL-Chain) and Utility-Occupancy-Table (UO-Table). Definitions about the proposed SUMU algorithm are given below.

\begin{definition}[$I$-Extension and $S$-Extension]
	\rm Extension \cite{han2001prefixspan, yin2012uspan, wang2016efficiently} is an operation that is often used to extend patterns in the pattern-growth-based algorithms. A pattern can be extended to derive a longer sup-pattern. There are two extension operations: one is called $I$-Extension and the other is called $S$-Extension. For a sequence $t$, through the $I$-Extension, a new item $i$ can be added to the last itemset of $t$. This extension operation about $t$ and $i$ can be denoted as $<$$t$ $\oplus$ $i$$>$. As for $S$-Extension, a new item $i$ is added to the last of $t$ as a new itemset. This extension operation about $t$ and $i$ can be denoted as $<$$t$ $\otimes$ $i$$>$.
\end{definition}

\subsection{UOL-Chain and UO-Table}

In this section, we design two compact data structures, including UOL-Chain and UO-Table. UOL-Chain is an extension of Utility-chain \cite{wang2016efficiently} to be able to deal with the problem of utility occupancy. As for UO-Table, it is a summary of important information from UOL-Chain. Based on these two data structures, not only the essential information about the patterns during the mining process is preserved, but also the identification of HUOSPs can be done quickly.

\begin{figure}[h]
	\centering
	\includegraphics[trim=0 0 0 0,clip,scale=0.3]{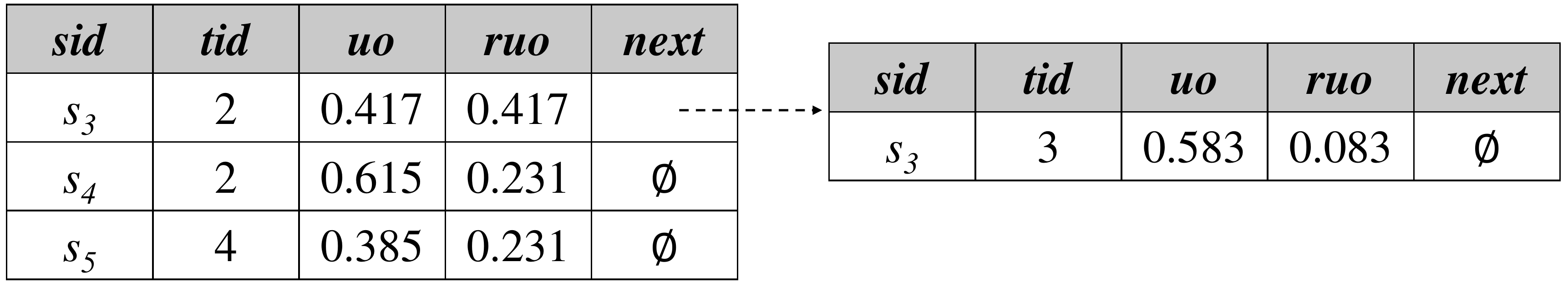}
	\caption{The {UOL-Chain} of the sequence $<$[$a$], [$c$]$>$}
	\label{UOL-Chain}
\end{figure}

The UOL-Chain of a sequence $t$ in the $q$-sequence $s$ has many indispensable elements, and Fig. \ref{UOL-Chain} shows the UOL-Chain of the sequence $<$[$a$], [$c$]$>$. The information carried by each element for $t$ in $s$ is as follows.

\begin{itemize}
	\item \textit{sid} is the unique identifier of the $q$-sequence $s$. 
	\item \textit{tid} is the $i$-th extension position.
	\item \textit{uo} is the utility occupancy of $t$ in $s$ at extension position.
	\item \textit{ruo} is the remaining utility occupancy of $t$ in $s$ at extension position.
	\item \textit{next} is the pointer that points to the next element (the ($i$+1)-th element).
\end{itemize}

\begin{figure}[h]
	\centering
	\includegraphics[trim=0 0 0 0,clip,scale=0.35]{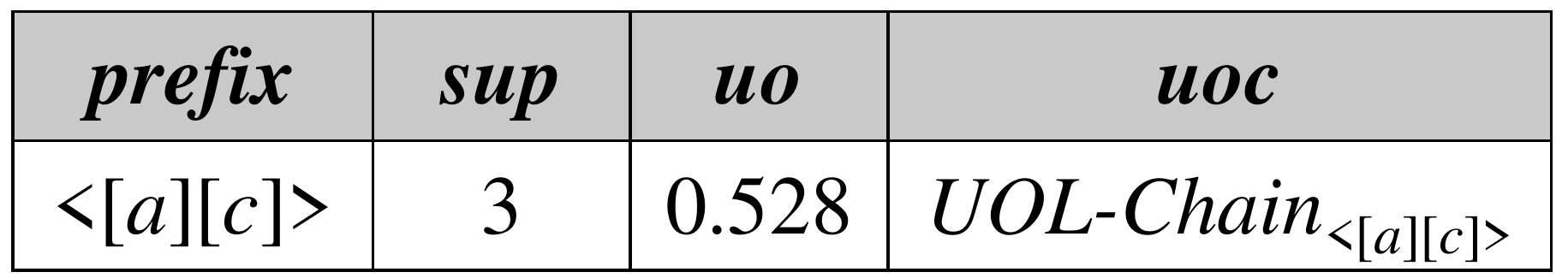}
	\caption{The {UO-Table} of the sequence $<$[$a$], [$c$]$>$}
	\label{UO-Table}
\end{figure}

Moreover, the UO-Table of the sequence $<$[$a$], [$c$]$>$ is shown in Fig. \ref{UO-Table}. It is clear that the support and utility occupancy of the sequence $<$[$a$], [$c$]$>$ are 3 and 0.528, respectively. Each element of a sequence carries the following information:

\begin{itemize}
	\item \textit{prefix} is the sequence. 
	\item \textit{sup} is the support of \textit{prefix} in the $q$-sequence database.
	\item \textit{uo} is the utility occupancy of \textit{prefix} in the $q$-sequence database.
	\item \textit{uoc} is the UOL-Chain of \textit{prefix} in the $q$-sequence database.
\end{itemize}

\subsection{Upper Bound on Utility Occupancy and Support}

In a utility-driven mining task, there are many upper bounds, including sequence-weighted utilization (\textit{SWU}) \cite{yin2012uspan}, sequence-utility upper bound (\textit{SUUB}) \cite{lan2014applying}, prefix extension utility (\textit{PEU}) \cite{wang2016efficiently}, etc., that are designed to terminate extension operations on unpromising candidate patterns early. Although some previous upper bounds are effective and efficient, they only work for HUSPM and are not suitable for mining HUOSPs. Such upper bounds need to be modified and improved so that they can be adapted to HUOSPM. However, due to the non-effectiveness of the \textit{downward} \textit{closure} property on utility occupancy, this makes designing powerful technologies an intractable challenge. With the guarantee of correctness and completeness, we introduce four upper bounds regarding utility occupancy on sequence data. In addition, contrary to the HUSPM, the support metric also needs to be evaluated in the HUOSPM. Therefore, with the aid of the anti-monotonicity of support, we introduce two upper bounds for support. Related theorems and proofs are provided.

\begin{upper bound}[Prefix Extension Utility Occupancy, \textit{PEUO}]
	\rm \textit{PEUO} is the first upper bound we propose, which is extended from \textit{PEU} \cite{wang2016efficiently}. For a sequence $t$, the \textit{PEUO} of it at a position $p$ of a $q$-sequence $s$, denoted as \textit{PEUO}($t$, $s$, $p$), is defined as follows:
	$$ \textit{PEUO}(t, s, p) =  \left\{
	\begin{aligned}
	\textit{uo}(t, s, p) + \textit{ruo}(t, s, p) &  & \textit{ruo}(t, s, p) \textgreater 0 \\
	0 & & otherwise
	\end{aligned}
	\right.
	$$
	We use $P$ to present a set containing all positions that a sequence $t$ matches. In a $q$-sequence $s$, the \textit{PEUO} of the sequence $t$, denoted as \textit{PEUO}($t$, $s$), is defined as $\textit{PEUO}(t, s)$ = $\textit{max}\{\textit{PEUO}(t, s, p) \;|\; p \in P\}$. Based on the above definitions, given a predefined \textit{minsup}, the total \textit{PEUO} of a sequence $t$ in a $q$-sequence database $\mathcal{D}$, denoted as \textit{PEUO}($t$), is defined as follows:
	$$
	\begin{aligned}
	\textit{PEUO}(t) = \frac{\sum\limits_{t \subseteq s \land s \in \mathcal{D}} \textit{PEUO}(t, s)}{\textit{minsup}}. 
	\end{aligned}
	$$
	
\end{upper bound}

\begin{theorem}
	\label{peuo:theorem}
	\rm Given two sequences $t$ and $t^\prime$ such that $t$ $\in$ $\mathcal{D}$ and $t^\prime$ $\in$ $\mathcal{D}$, and $t$ is a prefix of $t^\prime$. If \textit{sup}($t$) $\ge$ \textit{minsup} and \textit{sup}($t^\prime$) $\ge$ \textit{minsup}, it holds that:
	$$
	\begin{aligned}
	\textit{uo}(t^\prime) \le \textit{PEUO}(t) ,\; \textit{PEUO}(t^\prime) \le \textit{PEUO}(t). 
	\end{aligned}
	$$
\end{theorem}
\begin{proof}
	\rm In a $q$-sequence $s$, $p$ denotes the position that $t$ matches, and $p^\prime$ denotes the position that $t^\prime$ matches. For any $p$ and $p^\prime$, which satisfy $p$ less than $p^\prime$. According to the above-mentioned definitions, we can learn that:
	$$
	\begin{aligned}
	\textit{uo}(t^\prime) &= \frac{\sum\limits_{t^\prime \subseteq s \land s \in \mathcal{D}}\textit{uo}(t^\prime, s)}{\textit{sup}(t^\prime)} \\
	&= \frac{\sum\limits_{t \subset t^\prime \land t^\prime \subseteq s \land s \in \mathcal{D}}\textit{max}\{\textit{uo}(t, s, p) + \textit{uo}(i_{p^\prime}, s) \}}{\textit{sup}(t^\prime)} \\
	&\le  \frac{\sum\limits_{t \subset t^\prime \land t^\prime \subseteq s \land s \in \mathcal{D}}\textit{max}\{\textit{uo}(t, s, p) + \sum\limits_{p \textless j \le p^\prime}\textit{uo}(i_j, s)\}}{\textit{sup}(t^\prime)}\\
	&\le  \frac{\sum\limits_{t \subseteq s \land s \in \mathcal{D}}\textit{max}\{\textit{uo}(t, s, p) + \textit{ruo}(t, s, p)\}}{\textit{sup}(t^\prime)}\\
	&= \frac{\sum\limits_{t \subseteq s \land s \in \mathcal{D}}\textit{PEUO}(t, s)}{\textit{sup}(t^\prime)} \\
	&\le \frac{\sum\limits_{t \subseteq s \land s \in \mathcal{D}}\textit{PEUO}(t, s)}{\textit{minsup}} \\
&	= \textit{PEUO}(t).
	\end{aligned}
	$$
	For proofing \textit{PEUO}($t^\prime$) $\le$ \textit{PEUO}($t$), we can have that:
	$$
	\begin{aligned}
	\textit{PEUO}(t^\prime) &= \frac{\sum\limits_{t^\prime \subseteq s \land s \in \mathcal{D}}\textit{PEUO}(t^\prime, s)}{\textit{minsup}} \\
	&= \frac{\sum\limits_{t \subset t^\prime \land t^\prime \subseteq s \land s \in \mathcal{D}}\textit{max}\{\textit{uo}(t, s, p) + \textit{uo}(i_{p^\prime}, s) + \textit{ruo}(t, s, p^\prime)\}}{\textit{minsup}}\\
	&\le \frac{\sum\limits_{t \subset t^\prime \land t^\prime \subseteq s \land s \in \mathcal{D}}\textit{max}\{\textit{uo}(t, s, p) + \sum\limits_{p \textless j \le p^\prime}\textit{uo}(i_j, s) + \textit{ruo}(t, s, p^\prime)\}}{\textit{minsup}}\\
	&\le \frac{\sum\limits_{t \subseteq s \land s \in \mathcal{D}}\textit{max}\{\textit{uo}(t, s, p) + \textit{ruo}(t, s, p)\}}{\textit{minsup}} \\
	&= \frac{\sum\limits_{t \subseteq s \land s \in \mathcal{D}}\textit{PEUO}(t, s)}{\textit{minsup}}  \\
	&= \textit{PEUO}(t).
	\end{aligned}
	$$
\end{proof}

\begin{upper bound}[Reduced Sequence Utility Occupancy, \textit{RSUO}]
	\rm \textit{RSUO} is an upper bound designed for width pruning in pattern growth. For a sequence $t$, let $l$ be a sequence that is able to generate $t$ by performing an extension operation. In a $q$-sequence $s$, the \textit{RSUO} of the sequence $t$ can be denoted as \textit{RSUO}($t$, $s$) and is defined as follows:
	$$ \textit{RSUO}(t, s) =  \left\{
	\begin{aligned}
	\textit{PEUO}(l, s)  &  & l \subseteq s \land t \subseteq s\\
	0 & & otherwise
	\end{aligned}
	\right.
	$$
	Subsequently, given a predefined \textit{minsup}, the \textit{RSUO} of a sequence $t$ in a $q$-sequence database $\mathcal{D}$ can be denoted as \textit{RSUO}($t$) and defined as follows:
	$$
	\begin{aligned}
	\textit{RSUO}(t) = \frac{\sum\limits_{s \in \mathcal{D}} \textit{RSUO}(t, s)}{\textit{minsup}}. 
	\end{aligned}
	$$
\end{upper bound}

\begin{theorem}
	\label{rsuo:theorem}
	\rm Given two sequences $t$ and $t^\prime$ such that $t$ $\in$ $\mathcal{D}$ and $t^\prime$ $\in$ $\mathcal{D}$, and $t$ is a prefix of $t^\prime$, if \textit{sup}($t$) $\ge$ \textit{minsup} and \textit{sup}($t^\prime$) $\ge$ \textit{minsup}, it holds that: $	\textit{uo}(t^\prime) \le \textit{RSUO}(t)$; $\textit{RSUO}(t^\prime) \le \textit{RSUO}(t)$.
\end{theorem}

\begin{proof}
	\rm Let $l$ be a sequence that is able to generate a sequence $t$ by performing an extension operation. According to Theorem \ref{peuo:theorem}, we can learn that:
	$$
	\begin{aligned}
	\textit{uo}(t^\prime) &\le \frac{\sum\limits_{t \subseteq s \land s \in \mathcal{D}}\textit{PEUO}(t, s)}{\textit{minsup}}
	\le \frac{\sum\limits_{l \subseteq t \land t \subseteq s \land s \in \mathcal{D}}\textit{PEUO}(l, s)}{\textit{minsup}} = \frac{\sum\limits_{s \in \mathcal{D}} \textit{RSUO}(t, s)}{\textit{minsup}} = \textit{RSUO}(t).
	\end{aligned}
	$$ 

	$$
	\begin{aligned}
	\textit{RSUO}(t^\prime) &= \frac{\sum\limits_{s \in \mathcal{D}} \textit{RSUO}(t^\prime, s)}{\textit{minsup}}
	= \frac{\sum\limits_{t \subseteq s \land t^\prime \subseteq s \land s \in \mathcal{D}}\textit{PEUO}(t, s)}{\textit{minsup}}  \le \frac{\sum\limits_{l \subseteq s \land t \subseteq s \land t^\prime \subseteq s \land s \in \mathcal{D}}\textit{PEUO}(l, s)}{\textit{minsup}} = \textit{RSUO}(t).
	\end{aligned}
	$$
\end{proof}

\begin{upper bound}[Top Prefix Utility Occupancy, \textit{TPUO}]
	\rm Because the upper bound \textit{PEUO} is insufficient, we propose \textit{TPUO} as a more effective upper bound. \textit{TPUO} is based on the \textit{PEUO}, taking into account the more likely practical utility occupancy. It only calculates the \textit{PEUO} value for the top \textit{minsup} large (we use $\downarrow$ to present the descending order for \textit{PEUO} value), keeping the overestimation of the sequence of utility occupancy not significant. In a $q$-sequence $s$, the \textit{TPUO} of a sequence $t$, denoted as \textit{TPUO}($t$, $S$), is defined as follows: 
	$$
	\begin{aligned}
	\textit{TPUO}(t, s) = \textit{PEUO}(t, s).
	\end{aligned}
	$$
	Given a predefined \textit{minsup}, the \textit{TPUO} of a sequence $t$ in a $q$-sequence database $\mathcal{D}$, denoted as \textit{TPUO}($t$), is defined as follows: 
	$$
	\begin{aligned}
	\textit{TPUO}(t) = \frac{\sum\limits_{\textit{top minsup,}\; t \subseteq s \land s \in \mathcal{D}} \textit{TPUO}(t, s) \downarrow}{\textit{minsup}}. 
	\end{aligned}
	$$
\end{upper bound}

\begin{theorem}
	\label{tpuo:theorem}
	\rm Given two sequences $t$ and $t^\prime$ such that $t$ $\in$ $\mathcal{D}$ and $t^\prime$ $\in$ $\mathcal{D}$, and $t$ is a prefix of $t^\prime$, if \textit{sup}($t$) $\ge$ \textit{minsup} and \textit{sup}($t^\prime$) $\ge$ \textit{minsup}, it holds that: $
	\textit{uo}(t^\prime) \le \textit{TPUO}(t)$; $\textit{TPUO}(t^\prime) \le \textit{TPUO}(t)$.
\end{theorem}

\begin{proof}
	\rm For the sake of brevity, we use two variables, \textit{avg}$_1$ and \textit{avg}$_2$, which are defined as follows:
	$$
	\textit{avg}_1 = \frac{\sum\limits_{\textit{top minsup,}\;t^\prime \subseteq s \land s \in \mathcal{D}}\textit{uo}(t^\prime, s)\downarrow}{\textit{minsup}}
	$$
	$$
	\textit{avg}_2 = \frac{\sum\limits_{\textit{top minsup}+1 \sim \textit{sup}(t^\prime),\;t^\prime \subseteq s \land s \in \mathcal{D}}\textit{uo}(t^\prime, s)\downarrow}{\textit{sup}(t^\prime) - \textit{minsup}}
	$$
	Note that \textit{avg}$_1$ and \textit{avg}$_2$ satisfy \textit{avg}$_1$ $\ge$ \textit{avg}$_2$. According to the above-mentioned definitions and Theorem \ref{peuo:theorem}, we can learn that:
	$$
	\begin{aligned}
	\textit{uo}(t^\prime) &= \frac{\textit{avg}_1 \times \textit{minsup} + \textit{avg}_2 \times (\textit{sup}(t^\prime) - \textit{minsup})}{\textit{sup}(t^\prime)} \\
	&\le \frac{\textit{avg}_1 \times \textit{minsup} + \textit{avg}_1 \times (\textit{sup}(t^\prime) - \textit{minsup})}{\textit{sup}(t^\prime)} \\
	&= \textit{avg}_1
	= \frac{\sum\limits_{\textit{top minsup,}\;t^\prime \subseteq s \land s \in \mathcal{D}}\textit{uo}(t^\prime, s)\downarrow}{\textit{minsup}} \\
	&\le \frac{\sum\limits_{\textit{top minsup,}\;t \subseteq s \land s \in \mathcal{D}}\textit{TPUO}(t, s)\downarrow}{\textit{minsup}}
	= \textit{TPUO}(t).
	\end{aligned}
	$$
	
	$$
	\begin{aligned}
	\textit{TPUO}(t^\prime) &= \frac{\sum\limits_{\textit{top minsup,}\;t^\prime \subseteq s \land s \in \mathcal{D}}\textit{TPUO}(t^\prime, s)\downarrow}{\textit{minsup}}  \le \frac{\sum\limits_{\textit{top minsup,}\;t \subseteq s \land s \in \mathcal{D}}\textit{TPUO}(t, s)\downarrow}{\textit{minsup}} = \textit{TPUO}(t)
	\end{aligned}
	$$
\end{proof}

\begin{upper bound}[Top Sequence Utility Occupancy, \textit{TSUO}]
	\rm Similar to \textit{RSUO} designed for width pruning, we also design \textit{TSUO} on the basis of \textit{TPUO}. For a sequence $t$, let $l$ be a sequence that is able to generate $t$ by performing an extension operation. In a $q$-sequence $s$, the \textit{TSUO} of the sequence $t$ is denoted as \textit{TSUO}($t$, $s$) and defined as follows:
	$$ \textit{TSUO}(t, s) =  \left\{
	\begin{aligned}
	\textit{TPUO}(l, s)  &  & l \subseteq s \land t \subseteq s\\
	0 & & otherwise
	\end{aligned}
	\right.
	$$
	Given a predefined \textit{minsup}, the \textit{TSUO} of a sequence $t$ in a $q$-sequence database $\mathcal{D}$, denoted as \textit{TSUO}($t$), is defined as follows: 
	$$
	\begin{aligned}
	\textit{TSUO}(t) = \frac{\sum\limits_{\textit{top minsup,}\; s \in \mathcal{D}} \textit{TSUO}(t, s) \downarrow}{\textit{minsup}}. 
	\end{aligned}
	$$
\end{upper bound}

\begin{theorem}
	\label{tsuo:theorem}
	\rm Given two sequences $t$ and $t^\prime$ such that $t$ $\in$ $\mathcal{D}$ and $t^\prime$ $\in$ $\mathcal{D}$, and $t$ is a prefix of $t^\prime$, if \textit{sup}($t$) $\ge$ \textit{minsup} and \textit{sup}($t^\prime$) $\ge$ \textit{minsup}, it holds that: $
	\textit{uo}(t^\prime) \le \textit{TSUO}(t)$; $\textit{TSUO}(t^\prime) \le \textit{TSUO}(t)$.
\end{theorem}

\begin{proof}
	\rm Let $l$ be a sequence that is able to generate a sequence $t$ by performing an extension operation. According to Theorem \ref{tpuo:theorem}, we can learn that:
	$$
	\begin{aligned}
	\textit{uo}(t^\prime) &\le \frac{\sum\limits_{\textit{top minsup,}\;t \subseteq s \land s \in \mathcal{D}}\textit{TPUO}(t, s)\downarrow}{\textit{minsup}} \\
	&\le \frac{\sum\limits_{\textit{top minsup,}\; l \subseteq t \land t \subseteq s \land s \in \mathcal{D}}\textit{TPUO}(l, s)\downarrow}{\textit{minsup}} \\
	& = \frac{\sum\limits_{\textit{top minsup,}\; s \in \mathcal{D}} \textit{TSUO}(t, s) \downarrow}{\textit{minsup}} \\
	& = \textit{TSUO}(t)
	\end{aligned}
	$$ 

	$$
	\begin{aligned}
	\textit{TSUO}(t^\prime) &= \frac{\sum\limits_{\textit{top minsup,}\; s \in \mathcal{D}} \textit{TSUO}(t^\prime, s) \downarrow}{\textit{minsup}} \\
	&= \frac{\sum\limits_{\textit{top minsup,}\; t \subseteq s \land t^\prime \subseteq s \land s \in \mathcal{D}}\textit{TSUO}(t, s)\downarrow}{\textit{minsup}} \\
	&\le \frac{\sum\limits_{\textit{top minsup,}\; l \subseteq s \land t \subseteq s \land t^\prime \subseteq s \land s \in \mathcal{D}}\textit{TSUO}(l, s)\downarrow}{\textit{minsup}} \\
	& = \textit{TSUO}(t)
	\end{aligned}
	$$
\end{proof}

\begin{upper bound}[Prefix Extension Support, \textit{PES}]
	\rm To reduce meaningless pattern extensions during the mining process, we derive this upper bound based on \textit{PEUO}. For a sequence $t$, in a $q$-sequence $s$, it's \textit{PES} is denoted as \textit{PES}($t$, $s$) and defined as follows:
	$$ \textit{PES}(t, s) =  \left\{
	\begin{aligned}
	1 &  & \textit{PEUO}(t, s) \textgreater 0 \\
	0 & & otherwise
	\end{aligned}
	\right.
	$$
	The \textit{PES} of a sequence $t$ in a $q$-sequence database $\mathcal{D}$, denoted as \textit{PES}($t$), is defined as follows: 
	$$
	\begin{aligned}
	\textit{PES}(t) = \sum\limits_{t \subseteq s \land s \in \mathcal{D}} \textit{PES}(t, s). 
	\end{aligned}
	$$
\end{upper bound}

\begin{theorem}
	\label{pes:theorem}
	\rm	Given two sequences $t$ and $t^\prime$ such that $t$ $\in$ $\mathcal{D}$ and $t^\prime$ $\in$ $\mathcal{D}$, and $t$ is a prefix of $t^\prime$, it holds that: $	\textit{sup}(t^\prime) \le \textit{PES}(t)$; $\textit{PES}(t^\prime) \le \textit{PES}(t)$.
\end{theorem}
\begin{proof}
	\rm If $T$ is a set that contains all extended sequences of a sequence $t$, then $t^\prime$ $\in$ $T$. When \textit{PES}($t$) is not equal to 0, it indicates that a sequence can be extended from the sequence $t$. In this case, $T$ is not an empty set. Thus,
	$$
	\begin{aligned}
	\textit{sup}(t^\prime) = \textit{sup}(s, t^\prime \subseteq s \land s \in \mathcal{D})  \le \textit{sup}(s, \bigcup_{t^{\prime\prime} \in T \land t^{\prime\prime} \subseteq s \land s \in \mathcal{D}} s) = \textit{PES}(t).
	\end{aligned}
	$$

	$$
	\begin{aligned}
	\textit{PES}(t^\prime) = \textit{sup}(s,  \bigcup_{p \in P \land p \subseteq s \land s \in \mathcal{D}} s) \le \textit{sup}(s, \bigcup_{t^{\prime\prime} \in T \land t^{\prime\prime} \subseteq s \land s \in \mathcal{D}} s) = \textit{PES}(t),
	\end{aligned}
	$$
	where $T$ $\subseteq$ $P$.
\end{proof}

\begin{upper bound}[Reduced Sequence Support, \textit{RSS}]
	\rm For a sequence $t$, let $l$ be a sequence that is able to generate $t$ by performing an extension operation. In a $q$-sequence $s$, the \textit{RSS} of the sequence $t$, denoted as \textit{RSS}($t$, $s$), is defined as follows:
	$$ \textit{RSS}(t, s) =  \left\{
	\begin{aligned}
	\textit{PES}(l, s)  &  & l \subseteq s \land t \subseteq s\\
	0 & & otherwise
	\end{aligned}
	\right.
	$$
	The \textit{RSS} of a sequence $t$ in a $q$-sequence database $\mathcal{D}$, denoted as \textit{RSS}($t$), is defined as follows: 
	$$
	\begin{aligned}
	\textit{RSS}(t) = \sum\limits_{s \in \mathcal{D}} \textit{RSS}(t, s). 
	\end{aligned}
	$$
\end{upper bound}

\begin{theorem}
	\label{rss:theorem}
	\rm Given two sequences $t$ and $t^\prime$ such that $t$ $\in$ $\mathcal{D}$ and $t^\prime$ $\in$ $\mathcal{D}$, and $t$ is a prefix of $t^\prime$, it holds that: $	\textit{sup}(t^\prime) \le \textit{RSS}(t)$; $\textit{RSS}(t^\prime) \le \textit{RSS}(t)$.
\end{theorem}

\begin{proof}
	\rm Let $l$ be a sequence that is able to generate a sequence $t$ by performing an extension operation. According to Theorem \ref{pes:theorem}, we can learn that:
	$$
	\begin{aligned}
	\textit{sup}(t^\prime) \le \textit{PES}(t) \le \sum\limits_{l \subseteq t \land t \in S \land S \in \mathcal{D}}\textit{PES}(l, S) = \sum\limits_{S \in \mathcal{D}} \textit{RSS}(t, S) = \textit{RSS}(t).
	\end{aligned}
	$$

	$$
	\begin{aligned}
	\textit{RSS}(t^\prime) = \sum\limits_{S \in \mathcal{D}} \textit{RSS}(t^\prime, S) = \sum\limits_{t \subseteq S \land t^\prime \subseteq S \land S \in \mathcal{D}} \textit{PES}(t, S)  \le \sum\limits_{l \subseteq S \land t \subseteq S \land t^\prime \subseteq S \land S \in \mathcal{D}} \textit{PES}(l, S) = \textit{RSS}(t).
	\end{aligned}
	$$
\end{proof}

\subsection{Pruning Strategies for HUOSPM}

In this section, to reduce the candidate patterns and their operations, we design some pruning strategies based on the proposed upper bounds. These pruning strategies are based on the maximum utility occupancy and maximum support of the extended patterns. Here, the proposed pruning strategies \ref{strategy2}, \ref{strategy3}, \ref{strategy4}, and \ref{strategy5} are based on the consideration of utility occupancy, while the pruning strategies \ref{strategy1}, \ref{strategy6}, and \ref{strategy7} are concerned with support. 

\begin{strategy}
	\label{strategy1}
	\rm The task of HUOSPM is to mine those sequential patterns with high frequency and high utility occupancy. Based on the anti-monotonicity of support, we can filter those unpromising items whose frequency is less than \textit{minsup} after scanning the database. This process is safe because a superset pattern of an infrequent pattern must also be infrequent \cite{agrawal1995mining, han2001prefixspan}.
\end{strategy}

\begin{strategy}
	\label{strategy2}
	\rm According to Theorem \ref{peuo:theorem}, if the \textit{PEUO}($t$) is less than \textit{minuo}, we can safely terminate all extension operations of $t$. The utility occupancy of the descendant pattern with $t$ as a prefix is certainly not greater than \textit{minuo}. It is clear that this strategy is a depth pruning strategy.
\end{strategy}

\begin{strategy}
	\label{strategy3}
	\rm According to Theorem \ref{rsuo:theorem}, let $l$ be a sequence that is able to generate a sequence $t$ by performing an extension operation, if the \textit{RSUO}($t$) is less than \textit{minuo}, we can safely terminate this extension operation for $l$. The utility occupancy of $t$ and its descendant patterns are almost certainly less than \textit{minuo}. Note that this strategy is a width pruning strategy.
\end{strategy}

\begin{strategy}
	\label{strategy4}
	\rm According to Theorem \ref{tpuo:theorem}, if the \textit{TPUO}($t$) is less than \textit{minuo}, we can safely terminate all extension operations of $t$. The utility occupancy of the descendant pattern with $t$ as a prefix is almost certainly less than \textit{minuo}. Note that this strategy is a depth pruning strategy.
\end{strategy}

\begin{strategy}
	\label{strategy5}
	\rm According to Theorem \ref{tsuo:theorem}, let $l$ be a sequence capable of generating a sequence $t$ by performing an extension operation. If the \textit{TSUO}($t$) is less than \textit{minuo}, we can safely terminate this extension operation for $l$. The utility occupancy of $t$ and its  descendant patterns are certainly not greater than \textit{minuo}. Note that this strategy is a width pruning strategy.
\end{strategy}

\begin{strategy}
	\label{strategy6}
	\rm According to Theorem \ref{pes:theorem}, if the \textit{PES}($t$) is less than \textit{minsup}, we can safely terminate all extension operations for $t$. The support of the descendant pattern with $t$ as a prefix is certainly not greater than \textit{minsup}. Note that this strategy is a depth pruning strategy.
\end{strategy}

\begin{strategy}
	\label{strategy7}
	\rm According to Theorem \ref{rss:theorem}, let $l$ be a sequence capable of generating a sequence $t$ by performing an extension operation. If the \textit{RSS}($t$) is less than \textit{minsup}, we can safely terminate this extension operation for $l$. The support of $t$ and its descendant patterns are certainly not greater than \textit{minsup}. Note that this strategy is a width pruning strategy.
\end{strategy}

\subsection{Proposed SUMU Algorithm}

Using the designed compact data structures, the introduced upper bounds on support and utility occupancy, and the proposed pruning strategies, we propose the SUMU algorithm for mining all interesting HUOSPs. Inspired by the idea of pattern growth, our designed algorithm also adopts two extension operations and a projected database mechanism to grow patterns gradually.

\begin{algorithm}[ht]
	\caption{Proposed SUMU algorithm}
	\label{alg:SUMU}
	\LinesNumbered
	
	\KwIn{$\mathcal{D}$: a $q$-sequence database; \textit{utable}: an external utility table of all items in $\mathcal{D}$; \textit{minsup}: the minimum support threshold; \textit{minuo}: the minimum utility occupancy threshold.} 
	\KwOut{All interesting HUOSPs.}		
	
	scan $\mathcal{D}$ to calculate the utility of all $q$-sequence records and the support of all items in the database\;
	\For{$i$ $\in$ $\mathcal{D}$}{
		\If{\textit{sup}($i$) $\ge$ \textit{minsup}}{
			\textit{FS1} $\leftarrow$ \textit{FS1} $\cup$ $<$[$i$]$>$; \quad(\textbf{pruning strategy 1})\\
		}
	}
	remove those infrequent items from the $q$-sequence database $\mathcal{D}$\;
	build the UOL-Chain and UO-Table for all frequent items in \textit{FS1}\;
	\For{$<$[$i$]$>$ $\in$ \textit{FS1}}{
		\If{\textit{uo}($<$$i$$>$) $\ge$ \textit{minuo}}{
			\textit{HUOSPs} $\leftarrow$ \textit{HUOSPs} $\cup$ $<$[$i$]$>$\;
		}
		\If{\textit{sup}($<$[$i$]$>$) $\ge$ \textit{minsup}}{
			\textbf{call} \textit{\textbf{HOUSP-Search}}($<$[$i$]$>$, \textit{UO-Table}$_{<[i]>}$)\;
		}		
	}
	\textbf{return} \textit{HUOSPs}		
\end{algorithm}

Algorithm \ref{alg:SUMU} represents the main pseudocode of the proposed SUMU algorithm. Four parameters, including a $q$-sequence database $\mathcal{D}$, an external utility-table \textit{utable}, a predefined minimum support threshold \textit{minsup}, and a predefined minimum utility occupancy \textit{minuo} are its inputs. Notice that both \textit{minsup} and \textit{minuo} satisfy greater than 0 and less than or equal to 1. The $q$-sequence database is scanned by SUMU first to calculate the utility of all $q$-sequences and the support of all items appearing in the database (Line 1). It then utilizes the pruning strategy \ref{strategy1} to find out all frequent items and save them as sequences in a list \textit{FS1} (Lines 2-6). After that, the unpromising items are removed from the origin $\mathcal{D}$ (Line 7). The UOL-Chain and UO-Table of all frequent sequences in \textit{FS1} are built (Line 9). For each sequence in \textit{FS1}, with the help of its UO-Table, the algorithm can verify whether it is output as a HUOSP and whether it recursively calls the \textbf{HUOSP-Search} function for it (Lines 9-16). After all HUOSPs have been identified, they are returned (Line 17).

\begin{algorithm}[ht]
	\small
	\caption{HOUSP-Search procedure}
	\label{alg:HOUSP-Search}
	\LinesNumbered
	\KwIn{$t$: a sequence; \textit{UO-Table}$_{t}$: the {UO-Table} of $t$.} 
	\KwOut{All interesting HUOSPs.}	
		
	\If{\textit{PES}($t$) $\textless$ \textit{minsup}}{
		\textbf{return};  \quad(\textbf{pruning strategy 6})\\
	}
	\If{\textit{PEUO}($t$) (or \textit{TPUO}($t$)) $\textless$ \textit{minuo}}{
		\textbf{return};  \quad(\textbf{pruning strategy 2 / pruning strategy 4})\\
	}
	initialize \textit{IEList} $\leftarrow$ $\varnothing$, \textit{SEList} $\leftarrow$ $\varnothing$\;
	scan the projected database of $t$ to: \\
	\quad1. find items that can be performed with $I$-Extension and update \textit{IEList}\;
	\quad2. find items that can be performed with $S$-Extension and update \textit{SEList}\;
	\If{\textit{RRS}($i$) $\textless$ \textit{minsup}}{
		remove $i$ from \textit{IEList} (or \textit{SEList});  \quad(\textbf{pruning strategy 7})\\
	}
	\If{\textit{RSUO}($i$) (or \textit{TSUO}($i$)) $\textless$ \textit{minuo}}{
		remove $i$ from \textit{IEList} (or \textit{SEList});  \quad(\textbf{pruning strategy 3 / pruning strategy 5})\\
	}
	\For{$i$ $\in$ \textit{IEList}}{
		$t^\prime$ $\leftarrow$ $<$$t$ $\oplus$ $i$$>$\;
		build the {UOL-Chain} and {UO-Table} of $t^\prime$\;
		\If{\textit{sup}($t^\prime$) $\ge$ \textit{minsup}}{
			\If{\textit{uo}($t^\prime$) $\ge$ \textit{minuo}}{
				\textit{HUOSPs} $\leftarrow$ \textit{HUOSPs} $\cup$ $t^\prime$\;
			}
			\textbf{call} \textit{\textbf{HOUSP-Search}}($t^\prime$, \textit{UO-Table}$_{t^\prime}$)\;
		}
		
	}
	\For{$i$ $\in$ \textit{SEList}}{
		$t^\prime$ $\leftarrow$ $<$$t$ $\otimes$ $i$$>$\;
		build the {UOL-Chain} and {UO-Table} of $t^\prime$\;
		\If{\textit{sup}($t^\prime$) $\ge$ \textit{minsup}}{
			\If{\textit{uo}($t^\prime$) $\ge$ \textit{minuo}}{
				\textit{HUOSPs} $\leftarrow$ \textit{HUOSPs} $\cup$ $t^\prime$\;
			}
			\textbf{call} \textit{\textbf{HOUSP-Search}}($t^\prime$, \textit{UO-Table}$_{t^\prime}$)\;
		}
	}	
	% \textbf{return} \textit{HUOSPs}		
\end{algorithm}

The details of \textbf{HUOSP-Search} procedure are shown in Algorithm \ref{alg:HOUSP-Search}. A prefix sequence $t$ and its UO-Table are parameters in this procedure. It first utilizes the pruning strategy \ref{strategy6} on support and the pruning strategy \ref{strategy2} (or the pruning strategy \ref{strategy4}) on utility occupancy to terminate the related operations of some sequences that are unpromising to generate a HUOSP (Lines 1-6). Subsequently, two lists (\textit{IEList} and \textit{SEList}) are initialized (Line 7). They record which items can be performed in two extension operations. Similar to pattern-growth-based algorithms and based on the UOL-Chain of $t$, the search procedure then scans the projected database of $t$ to identify those items that can be performed with the $I$-Extension or the $S$-Extension (Lines 8-10). For \textit{IEList} and \textit{SEList}, there are some hash tables associated with them that need to be constructed as well. Those hash tables record the information related to utility occupancy and support in order to perform the subsequent filtering. To reduce the following operations of unnecessary data structure construction, the pruning strategies \ref{strategy7} and \ref{strategy3} (or the pruning strategy \ref{strategy5}) are utilized to remove those unpromising items from \textit{IEList} and \textit{SEList} (Lines 11-16). After that, the sequence $t$ and the items of \textit{IEList} are performed by the $I$-Extension to generate a new sequence $t^\prime$; the UOL-Chain and UO-Table for $t^\prime$ are also constructed (Lines 18-19). Following the identification process, the involved recursive function decides whether or not to call based on the support and utility occupancy of $t^\prime$ (Lines 20-25). Related operations on the items within \textit{SEList} are similar (Lines 27-36). The pattern grows progressively as the recursive function is continuously called. In the end, all HUOSPs are mined.

\section{Experiments}  \label{sec:experiments}

We selected both real and synthetic datasets to conduct related experiments. The proposed SUMU algorithm is the first approach to mining sequential patterns with a utility occupancy measure. Thus, there is no suitable algorithm for comparison. We mainly focused on verifying the efficiency of the proposed upper bounds and pruning strategies and the effectiveness of SUMU. The code related to SUMU is programmed using the Java language and developed in Eclipse. Our extensive experiments were conducted on a bare computer, which is equipped with an i7-12700F 2.10 GHz CPU and 16 GB of RAM. The experimental details and results are shown below.

\subsection{Experimental Setup and Datasets}

Three real datasets (including Bible, FIFA, and Sign) and three synthetic datasets (including Syn10k, Syn20k, and Syn40k) were used in the experiments. The real datasets are often used in the evaluation of pattern mining algorithms and can be accessed from the website SPMF\footnote{\url{http://www.philippe-fournier-viger.com/spmf/}}. Regarding the used synthetic datasets, they can be generated by the IBM Quest Synthetic Data Generator \cite{QSD}. Each dataset has its own characteristics and can represent a specific type of data in practical applications. The characteristics of these datasets are described below.

$ \bullet $ \textit{\textbf{Bible}} contains 13,905 items and 36,369 sequences, which are transformed from the book Bible. Its average sequence length is 21.64.

$ \bullet $ \textit{\textbf{FIFA}} contains 2,990 items and 20,450 sequences  derived from the website of FIFA World Cup 98. Its average sequence length is 36.23.

$ \bullet $ \textit{\textbf{Sign}} a small but dense dataset of sign language utterances, with 267 items and 730 sequences. Its average sequence length is 27.11.

$ \bullet $ \textit{\textbf{Syn10k}} is a synthetic dataset with 10,000 sequence records. It has 7,312 distinct items, and its average sequence length is 26.97.

$ \bullet $ \textit{\textbf{Syn20k}} is a synthetic dataset with 20,000 sequence records. It has 7,442 distinct items, and its average sequence length is 26.84.

$ \bullet $ \textit{\textbf{Syn40k}} is a synthetic dataset with 40,000 sequence records. It has 7,537 distinct items, and its average sequence length is 26.84.

To better evaluate the proposed SUMU algorithm, several variants regarding SUMU have also been designed. Therefore, the experimental results can better show the capabilities of the designed upper bounds and pruning strategies. In our experiments, the proposed SUMU algorithm with upper bounds \textit{PEUO} and \textit{RSUO} is denoted as SUMU$_\textit{simple}$. It means that only Strategies \ref{strategy2} and \ref{strategy3} are used in SUMU$_\textit{simple}$. On the basis of SUMU$_\textit{simple}$, if unpromising items are filtered out (with Strategy \ref{strategy1}) before generating HUOSPs, then this variant of SUMU is denoted as SUMU$_\textit{PEUO}$. To analyze the performance of gap between \textit{PEUO} and \textit{TPUO}, between \textit{RSUO} and \textit{TSUO} in the experiments, we also designed another variant of SUMU (with Strategies \ref{strategy1}, \ref{strategy4}, and \ref{strategy5}), denoted as SUMU$_\textit{TPUO}$. In addition, on the basis of SUMU$_\textit{PEUO}$, the fourth variant, namely SUMU$_\textit{PES}$, is designed to evaluate the two upper bounds on the support measure. Those variants of SUMU are compared to comprehensively evaluate the effectiveness and efficiency of SUMU.

\subsection{Pattern Analysis}

In this section, we mainly discuss the effect of the change in the number of HUOSPs as the \textit{minsup} or \textit{minsuo} changes. The results for various \textit{minsup} and under a fixed \textit{minuo} are shown in Table \ref{patterns_minsup}. Likewise, the results for various \textit{minuo} and under a fixed \textit{minsup} are shown in Table \ref{patterns_minuo}. For each dataset, we use \textit{minsup}$_1$, \textit{minsup}$_2$ (or \textit{minuo}$_1$, \textit{minuo}$_2$), and so on to indicate that we increasingly adjust the parameter \textit{minsup} (or \textit{minuo}). For instance, in our experiments, for the Bible dataset, the six parameters on \textit{minsup} are set to 300, 400, 500, 600, 700, and 800; and the six parameters on \textit{minuo} are set to 0.01, 0.03, 0.05, 0.07, 0.09, and 0.11. The detailed parameter settings can be observed in Fig. \ref{runtime_minsup} and Fig. \ref{runtime_minuo}.

% Please add the following required packages to your document preamble:
% \usepackage{multirow}
\begin{table}[H]
	\centering
	\caption{Number of patterns generated by varying \textit{minsup}}
	\label{patterns_minsup}
	\resizebox{\columnwidth}{!}{
		\begin{tabular}{c|cccccc}
			\hline \hline
			\multirow{2}{*}{\textbf{Dataset}}   & \multicolumn{6}{c}{\# \textbf{patterns}} \\ \cline{2-7} 
			&          $\textit{minsup}_{1}$         &  $\textit{minsup}_{2}$ & $\textit{minsup}_{3}$  & $\textit{minsup}_{4}$  & $\textit{minsup}_{5}$  & $\textit{minsup}_{6}$    \\ \hline
			Bible, \textit{minuo} = 0.1   & 21,442 & 11,008 & 6,527 & 4,290 & 2,993 & 2,211  \\ \hline
			FIFA, \textit{minuo} = 0.1 & 1,162 & 259 & 87 & 38 & 14 & 7    \\ \hline
			Sign, \textit{minuo} = 0.1 & 147,517 & 74,532 & 40,936 & 23,879 & 14,521 & 9,165  \\ \hline
			Syn10k, \textit{minuo} = 0.1 &   5,732,182 & 1,311,583 & 488,651 & 165,915 & 96,636 & 76,824  \\ \hline
			Syn20k, \textit{minuo} = 0.1 &   3,751,369 & 1,470,986 & 766,501 & 325,895 & 178,157 & 124,254 \\ \hline
			Syn40k, \textit{minuo} = 0.1&  7,223,421 & 5,144,928 & 3,710,872 & 2,087,673 & 1,142,202 & 770,393 \\ \hline
			\hline 
		\end{tabular}
	}
\end{table}

\begin{table}[H]
	\centering
	\caption{Number of patterns generated by varying \textit{minuo}}
	\label{patterns_minuo}
	\resizebox{\columnwidth}{!}{
		\begin{tabular}{c|cccccc}
			\hline \hline
			\multirow{2}{*}{\textbf{Dataset}}   & \multicolumn{6}{c}{\# \textbf{patterns}} \\ \cline{2-7} 
			&          $\textit{minuo}_{1}$         &  $\textit{minuo}_{2}$ & $\textit{minuo}_{3}$  & $\textit{minuo}_{4}$  & $\textit{minuo}_{5}$  & $\textit{minuo}_{6}$    \\ \hline
			{Bible}, \textit{minsup} = 500   & 11721 & 11668 & 11390 & 10433 & 8037 & 5012
			\\ \hline
			FIFA, \textit{minsup} = 4,000 & 1,093 & 870 & 499 & 212 & 69 & 20  \\ \hline
			Sign, \textit{minsup} = 70 & 40,936 & 28,375 & 18,330 & 11,087 & 6,134 & 3,136  \\ \hline
			Syn10k, \textit{minsup} = 14 & 488,651 & 435,881 & 367,058 & 287,760 & 204,788 & 130,351   \\ \hline
			Syn20k, \textit{minsup} = 24 & 766,501 & 660,716 & 513,359 & 355,981 & 217,649 & 117,495  \\ \hline
			Syn40k, \textit{minsup} = 34 & 7,223,421 & 6,737,579 & 5,831,242 & 4,550,068 & 3,092,664 & 1,777,904  \\ \hline
			\hline 
		\end{tabular}
	}
\end{table}

From Tables \ref{patterns_minsup} and \ref{patterns_minuo}, it is clear that the number of generated HUOSPs on each dataset is quite different as \textit{minsup} or \textit{minuo} is adjusted. Particularly, the number of generated HUOSPs on the synthetic datasets is higher than that on the real datasets. This is because for these synthetic datasets, each of their itemsets contains multiple items and can format more candidate patterns. Furthermore, as \textit{minsup} decreases by interval, the number of HUOSPs increases rapidly. For example, the difference between the number of patterns generated by \textit{minsup}$_1$ and the number of patterns generated under \textit{minsup}$_2$ is smaller than the difference between \textit{minsup}$_2$ and \textit{minsup}$_1$. This phenomenon is reasonable and also occurs in frequent itemset mining or sequential pattern mining. However, this phenomenon is contrary to the utility occupancy measure. The number of generated HUOSPs gradually increases as \textit{minuo} is decreased. This is because the HUOSPs generated by the algorithm SUMU do not vary much for smaller \textit{minuo} settings. In fact, such similar situations also can be found in the HUOPM algorithm \cite{gan2019huopm}. 

\subsection{Efficiency Analysis}

In this subsection, we conducted extensive experiments to evaluate the performance of the different upper bounds and pruning strategies used in SUMU. The results in terms of runtime for various \textit{minsup} and \textit{minuo} settings are shown in Fig. \ref{runtime_minsup} and Fig. \ref{runtime_minuo}. And the results in terms of candidate patterns for various \textit{minsup} and \textit{minuo} settings are shown in Tables \ref{candidates_minsup} and \ref{candidates_minuo}.

\begin{figure}[h]
	\centering
	\includegraphics[trim=0 0 0 0,clip,scale=0.32]{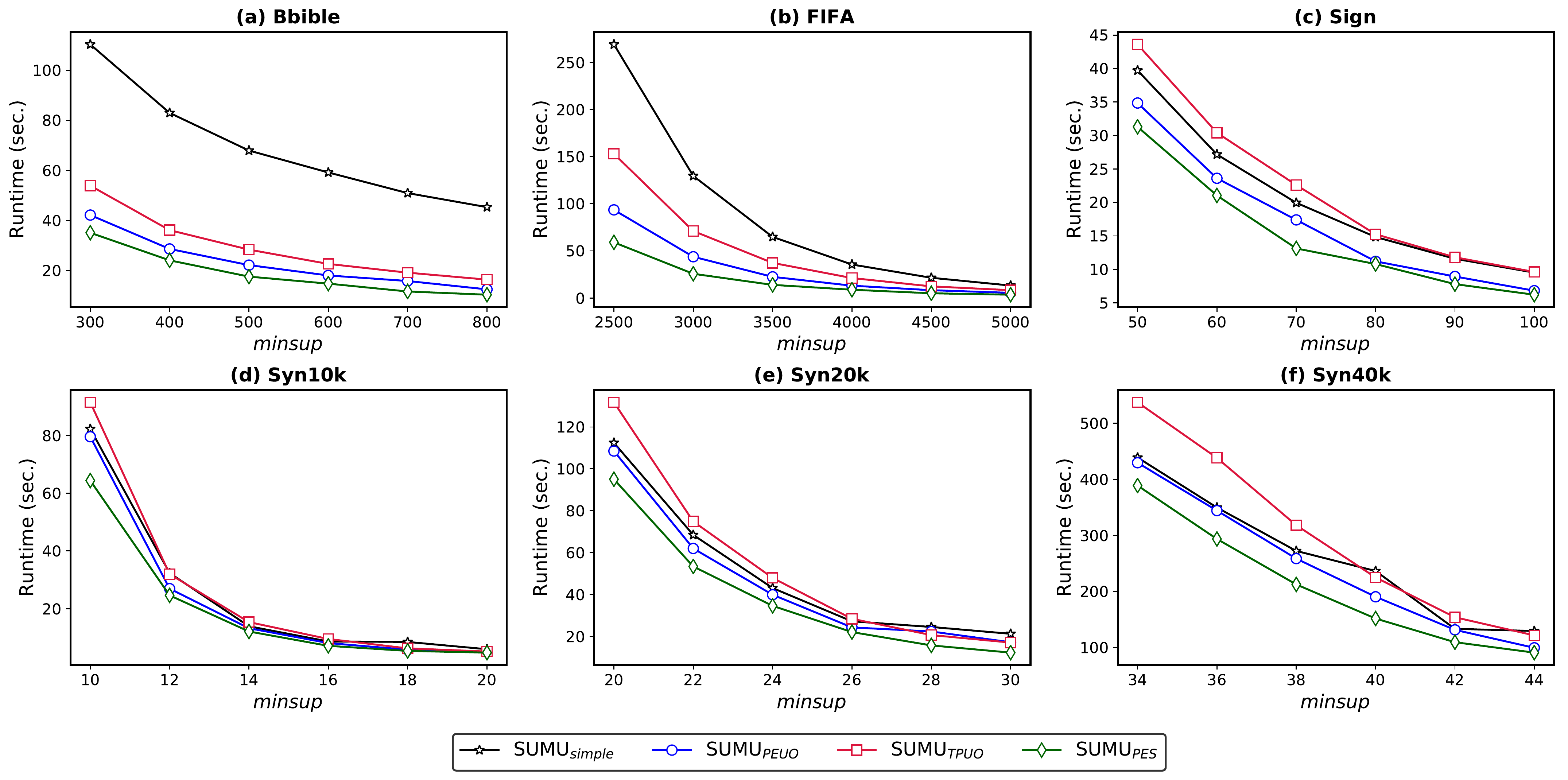}
	\caption{Running time under various \textit{minsup} and a fixed \textit{minuo} = 0.1.}
	\label{runtime_minsup}
\end{figure}

\begin{table}[h]
	\centering
	\caption{Number of candidate patterns generated by varying \textit{minsup}}
	\label{candidates_minsup}
	\resizebox{\columnwidth}{!}{
		\begin{tabular}{|c|c|c|c|c|c|c|c|}
			\hline \textbf{Dataset} & \textbf{Result} & $\textit{minsup}_{1}$         &  $\textit{minsup}_{2}$ & $\textit{minsup}_{3}$  & $\textit{minsup}_{4}$  & $\textit{minsup}_{5}$  & $\textit{minsup}_{6}$ \\
			\hline  \hline
			\multirow{4}{*}{\shortstack{Bible\\ \textit{minuo} = 0.1}} 
			& {SUMU$_\textit{simple}$} & 464,804 & 263,807 & 171,872 & 123,398 & 94,177 & 75,449 \\
			\cline{2-2}
			& {SUMU$_\textit{PEUO}$} & 321,866 & 168,372 & 103,179 & 68,755 & 48,108 & 35,632 \\
			\cline{2-2}
			& {SUMU$_\textit{TPUO}$} & 321,866 & 168,372 & 103,179 & 68,755 & 48,108 & 35,632 \\
			\cline{2-2}
			& {SUMU$_\textit{PES}$} & 35,367 & 18,999 & 11,721 & 7,967 & 5,737 & 4,354 \\
			\hline
			
			\multirow{4}{*}{\shortstack{FIFA \\ \textit{minuo} = 0.1}}
			& {SUMU$_\textit{simple}$} & 678,816 & 268,081 & 115,283 & 57,603 & 31,384 & 18,504 \\
			\cline{2-2}
			& {SUMU$_\textit{PEUO}$} & 214,531 & 80,469 & 35,373 & 17,818 & 9,930 & 5,545 \\
			\cline{2-2}
			& {SUMU$_\textit{TPUO}$} & 214,531 & 80,469 & 35,373 & 17,818 & 9,930 & 5,545 \\
			\cline{2-2}
			& {SUMU$_\textit{PES}$} & 14,710 & 5,787 & 2,399 & 1,099 & 557 & 296 \\
			\hline

			\multirow{4}{*}{\shortstack{Sign\\ \textit{minuo} = 0.1}} 
			& {SUMU$_\textit{simple}$} & 4,237,763 & 2,494,589 & 1,588,257 & 1,061,989 & 742,966 & 538,042 \\
			\cline{2-2}
			& {SUMU$_\textit{PEUO}$} & 3,668,153 & 2,131,189 & 1,284,130 & 834,756 & 553,160 & 390,477 \\
			\cline{2-2}
			& {SUMU$_\textit{TPUO}$} & 3,668,153 & 2,131,189 & 1,284,130 & 834,756 & 553,160 & 390,477 \\
			\cline{2-2}
			& {SUMU$_\textit{PES}$} & 372,610 & 208,839 & 126,752 & 81,340 & 54,695 & 38,095 \\
			\hline
			
			\multirow{4}{*}{\shortstack{Syn10k \\ \textit{minuo} = 0.1}} 
			& {SUMU$_\textit{simple}$} & 34,672,439 & 10,131,127 & 4,119,006 & 1,870,268 & 1,126,628 & 802,480 \\
			\cline{2-2}
			& {SUMU$_\textit{PEUO}$} & 32,762,145 & 9,533,071 & 3,727,340 & 1,661,839 & 974,145 & 725,941 \\
			\cline{2-2}
			& {SUMU$_\textit{TPUO}$} & 32,762,136 & 9,533,068 & 3,727,339 & 1,661,839 & 974,145 & 725,941 \\
			\cline{2-2}
			& {SUMU$_\textit{PES}$} & 5,968,170 & 1,412,210 & 537,899 & 194,708 & 115,473 & 89,559 \\
			\hline
			
			\multirow{4}{*}{\shortstack{Syn20k \\ \textit{minuo} = 0.1}} 
			& {SUMU$_\textit{simple}$} & 23,741,985 & 11,371,864 & 5,991,112 & 3,132,183 & 2,004,809 & 1,497,083 \\
			\cline{2-2}
			& {SUMU$_\textit{PEUO}$} & 23,391,275 & 11,167,493 & 5,706,748 & 2,943,750 & 1,845,769 & 1,364,875 \\
			\cline{2-2}
			& {SUMU$_\textit{TPUO}$} & 23,391,272 & 11,167,489 & 5,706,744 & 2,943,747 & 1,845,768 & 1,364,875 \\
			\cline{2-2}
			& {SUMU$_\textit{PES}$} & 3,916,112 & 1,578,521 & 841,369 & 377,140 & 215,095 & 151,514 \\
			\hline
			
			\multirow{4}{*}{\shortstack{Syn40k \\ \textit{minuo} = 0.1}} 
			& {SUMU$_\textit{simple}$} & 39,435,204 & 29,989,061 & 21,471,028 & 14,398,082 & 8,602,677 & 6,268,961 \\
			\cline{2-2}
			& {SUMU$_\textit{PEUO}$} & 38,953,081 & 29,551,866 & 21,110,384 & 14,105,707 & 8,252,704 & 5,995,403 \\
			\cline{2-2}
			& {SUMU$_\textit{TPUO}$} & 38,953,054 & 29,551,845 & 21,110,361 & 14,105,693 & 8,252,697 & 5,995,389 \\
			\cline{2-2}
			& {SUMU$_\textit{PES}$} & 7,454,317 & 5,326,206 & 3,857,950 & 2,206,183 & 1,239,766 & 850,635 \\
			\hline
			
			\hline
		\end{tabular}
	}
	
\end{table}

From Fig \ref{runtime_minsup} and Table \ref{candidates_minsup}, under different \textit{minsup} settings, we can clearly see that the runtime of the variant SUMU$_\textit{simple}$ is the worst on the datasets Bible and FIFA; the runtime of the variant SUMU$_\textit{TPUO}$ is the worst on the datasets Sign, Syn20k, and Syn40k. SUMU$_\textit{simple}$ and SUMU$_\textit{TPUO}$ perform similarly on Syn10k. However, when \textit{minsup} is set to 10, the runtime of SUMU$_\textit{TPUO}$ exceeds that of SUMU$_\textit{simple}$. In addition, the variant SUMU$_\textit{PES}$ which uses four upper bounds (\textit{PEUO}, \textit{RSUO}, \textit{PES}, and \textit{RSS}) can achieve the best performance on all datasets. And the variant SUMU$_\textit{PEUO}$ takes the second least amount of runtime. SUMU$_\textit{PES}$ is able to minimize candidate pattern generation, while SUMU$_\textit{simple}$ is the least reduced. The results of our experiment are as expected. In experiments under different \textit{minsup} and a fixed \textit{minuo}, we can draw the following conclusions.

\begin{enumerate}[label=(\arabic*)]
	\item SUMU$_\textit{PES}$ adopts a sufficient number of pruning strategies to significantly reduce candidate patterns while achieving the shortest  runtime. Compared to several other variants of SUMU, SUMU$_\textit{PES}$ generates much fewer candidate patterns. Although the number of candidate patterns is many times less than the other variants, the overall performance is not more than a few times better. This is a lot of unpromising candidate patterns that are also ignored in the subsequent program steps.
	
	\item The difference between SUMU$_\textit{simple}$ and SUMU$_\textit{PEUO}$ demonstrate that the pruning Strategy \ref{strategy1} is ineffective on synthetic datasets. This is because \textit{minsup} are set to relatively small values, and thus there are not many unpromising items appearing in the sequence dataset.
	
	\item Although \textit{TPUO} and \textit{TSUO} are tighter upper bounds, their calculation makes SUMU$_\textit{TPUO}$ take longer than SUMU$_\textit{PEUO}$. For a candidate pattern, SUMU$_\textit{PEUO}$ is able to compute upper bounds \textit{PEUO} and \textit{RSUO} in a liner time. While for \textit{TSUO}, it requires multiple sorting operations, which is a complex process. In addition, SUMU$_\textit{TPUO}$ does not reduce any candidate pattern on many datasets (including Bible, FIFA, and Sign). Even if it works on the few remaining datasets, it only reduces the number of candidate patterns by a particularly small amount.
\end{enumerate}

\begin{figure}[h]
	\centering
	\includegraphics[trim=0 0 0 0,clip,scale=0.32]{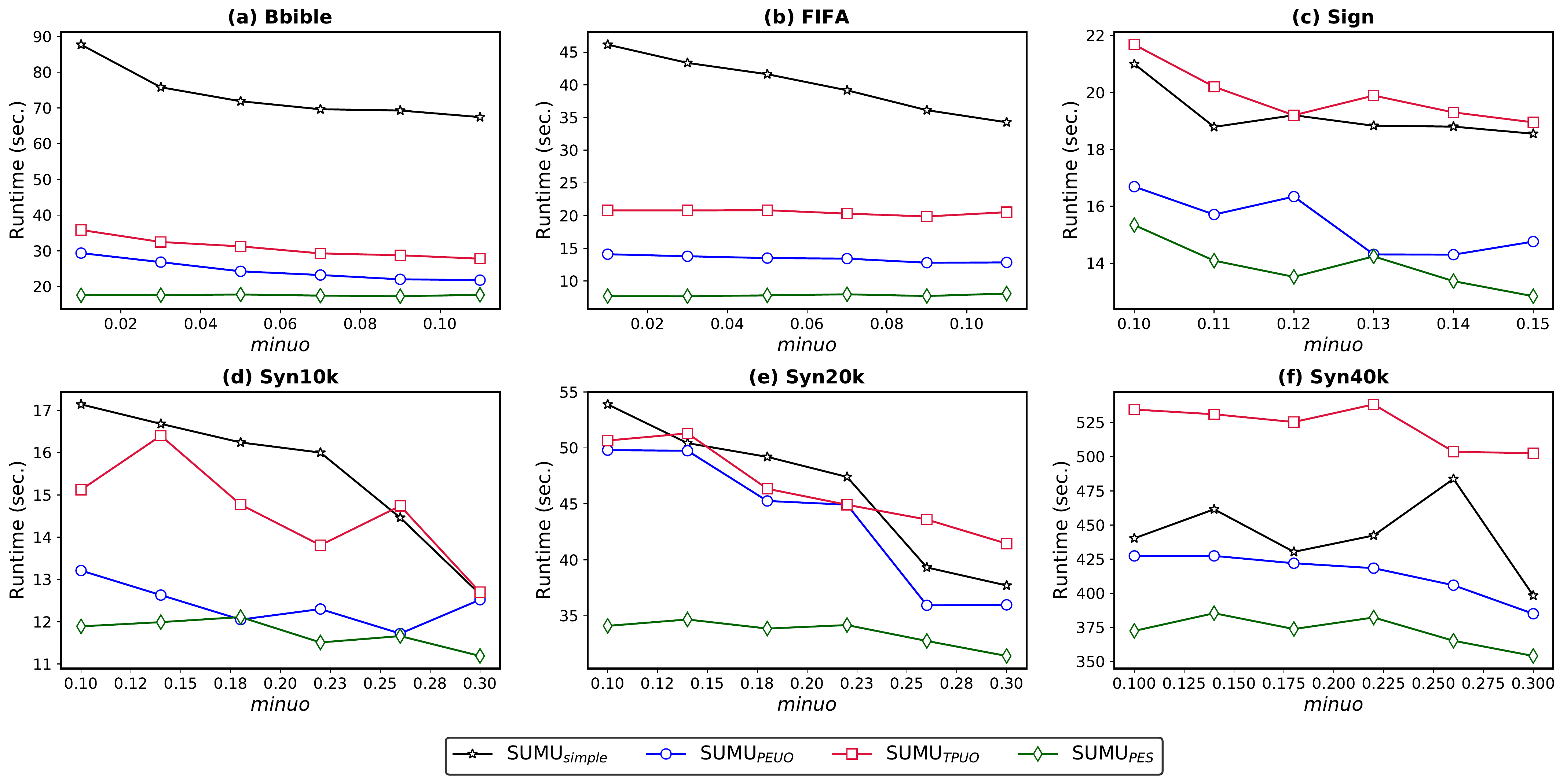}
	\caption{Running time under various \textit{minuo} and a fixed \textit{minsup}. (a) Bible, \textit{minsup} = 500. (b) FIFA,  \textit{minsup} = 4,000. (c) Sign, \textit{minsup} = 70. (d) Syn10k, \textit{minsup} = 14. (e) Syn20k, \textit{minsup} = 24. (f) Syn40k, \textit{minsup} = 34.}
	\label{runtime_minuo}
\end{figure}

\begin{table}[h]
	\centering
	\caption{Number of candidate patterns generated by varying \textit{minuo}}
	\label{candidates_minuo}
	\resizebox{\columnwidth}{!}{
		\begin{tabular}{|c|c|c|c|c|c|c|c|}
			\hline \textbf{Dataset} & \textbf{Result} & $\textit{minuo}_{1}$         &  $\textit{minuo}_{2}$ & $\textit{minuo}_{3}$  & $\textit{minuo}_{4}$  & $\textit{minuo}_{5}$  & $\textit{minuo}_{6}$ \\
			\hline  \hline
			\multirow{4}{*}{\shortstack{Bible\\ \textit{minsup} = 500}} 
			& {SUMU$_\textit{simple}$} & 1,921,104 & 672,801 & 390,245 & 264,424 & 195,445 & 153,171 \\
			\cline{2-2}
			& {SUMU$_\textit{PEUO}$} & 1,102,717 & 443,242 & 248,731 & 162,948 & 117,785 & 91,903 \\
			\cline{2-2}
			& {SUMU$_\textit{TPUO}$} & 1,102,717 & 443,242 & 248,731 & 162,948 & 117,785 & 91,903 \\
			\cline{2-2}
			& {SUMU$_\textit{PES}$} & 11,721 & 11,721 & 11,721 & 11,721 & 11,721 & 11,721 \\
			\hline
			
			\multirow{4}{*}{\shortstack{FIFA \\ \textit{minsup} = 4000}}
			& {SUMU$_\textit{simple}$} & 280,551 & 178,428 & 144,686 & 102,302 & 69,806 & 48,009 \\
			\cline{2-2}
			& {SUMU$_\textit{PEUO}$} & 38,098 & 30,487 & 25,968 & 22,128 & 19,040 & 16,766 \\
			\cline{2-2}
			& {SUMU$_\textit{TPUO}$} & 38,098 & 30,487 & 25,968 & 22,128 & 19,040 & 16,766 \\
			\cline{2-2}
			& {SUMU$_\textit{PES}$} & 1,099 & 1,099 & 1,099 & 1,099 & 1,099 & 1,099 \\
			\hline

			\multirow{4}{*}{\shortstack{Sign\\ \textit{minsup} = 70}} 
			& {SUMU$_\textit{simple}$} & 1,588,257 & 1,393,525 & 1,233,358 & 1,101,274 & 988,726 & 894,110 \\
			\cline{2-2}
			& {SUMU$_\textit{PEUO}$} & 1,284,130 & 1,129,929 & 1,002,217 & 895,976 & 806,047 & 730,904 \\
			\cline{2-2}
			& {SUMU$_\textit{TPUO}$} & 1,284,130 & 1,129,928 & 1,002,216 & 895,923 & 805,911 & 730,644 \\
			\cline{2-2}
			& {SUMU$_\textit{PES}$} & 126,752 & 126,743 & 126,716 & 126,652 & 126,478 & 126,240 \\
			\hline
			
			\multirow{4}{*}{\shortstack{Syn10k \\ \textit{minsup} = 14}} 
			& {SUMU$_\textit{simple}$} & 4,119,006 & 2,916,190 & 2,399,859 & 2,021,727 & 1,747,076 & 1,534,885 \\
			\cline{2-2}
			& {SUMU$_\textit{PEUO}$} & 3,727,340 & 2,790,561 & 2,301,336 & 1,948,286 & 1,680,334 & 1,485,466 \\
			\cline{2-2}
			& {SUMU$_\textit{TPUO}$} & 3,727,339 & 2,790,547 & 2,301,264 & 1,948,081 & 1,679,416 & 1,482,724 \\
			\cline{2-2}
			& {SUMU$_\textit{PES}$} & 537,899 & 537,774 & 537,328 & 536,225 & 533,228 & 526,147 \\
			\hline
			
			\multirow{4}{*}{\shortstack{Syn20k \\ \textit{minsup} = 24}} 
			& {SUMU$_\textit{simple}$} & 5,991,112 & 5,044,232 & 4,517,900 & 4,112,343 & 3,773,772 & 3,423,454 \\
			\cline{2-2}
			& {SUMU$_\textit{PEUO}$} & 5,706,748 & 4,904,484 & 4,430,131 & 4,041,273 & 3,712,031 & 3,357,823 \\
			\cline{2-2}
			& {SUMU$_\textit{TPUO}$} & 5,706,744 & 4,904,420 & 4,429,752 & 4,040,166 & 3,709,467 & 3,352,790 \\
			\cline{2-2}
			& {SUMU$_\textit{PES}$} & 841,369 & 840,970 & 838,452 & 830,080 & 812,793 & 785,422 \\
			\hline
			
			\multirow{4}{*}{\shortstack{Syn40k \\ \textit{minsup} = 34}} 
			& {SUMU$_\textit{simple}$} & 39,435,204 & 36,517,396 & 34,219,382 & 31,977,235 & 29,399,324 & 26,176,962 \\
			\cline{2-2}
			& {SUMU$_\textit{PEUO}$} & 38,953,081 & 36,148,568 & 33,947,656 & 31,690,858 & 29,135,445 & 25,872,139 \\
			\cline{2-2}
			& {SUMU$_\textit{TPUO}$} & 38,953,054 & 36,148,106 & 33,944,569 & 31,678,889 & 29,102,992 & 25,803,949 \\
			\cline{2-2}
			& {SUMU$_\textit{PES}$} & 7,454,317 & 7,452,659 & 7,440,512 & 7,403,338 & 7,327,730 & 7,200,566 \\
			\hline
			
			\hline
		\end{tabular}
	}
	
\end{table}

Furthermore, from Fig \ref{runtime_minuo} and Table \ref{candidates_minuo}, under different \textit{minuo} settings, we can clearly observe that SUMU$_\textit{PES}$ is the fastest variant of SUMU. On the datasets Sign, Syn10k, Syn20k, and Syn40k, there are some fluctuations that occur in all the variants of SUMU, but the overall trend is still clear. Regardless of which dataset is processed, the runtime curve of SUMU$_\textit{PES}$ becomes smoother as \textit{minuo} is adjusted. In particular, the number of candidate patterns generated on the Bible and FIFA datasets has not changed. In experiments under different \textit{minuo} and a fixed \textit{minsup}, we can draw the following conclusions.

\begin{enumerate}[label=(\arabic*)]
	\item Unlike the experiments under tuning \textit{minsup}, the runtime of each variant of SUMU is not much affected by the setting of \textit{minuo}. On the datasets Bible and FIFA, the runtimes of SUMU$_\textit{PEUO}$, SUMU$_\textit{TPUO}$, and SUMU$_\textit{PES}$ hardly increase when \textit{minuo} decreases. While the candidate patterns for SUMU$_\textit{PEUO}$ and SUMU$_\textit{TPUO}$ are increased substantially. This suggests that the support measure plays a greater role in determining the program runtime than the utility occupancy measure.
	
	\item SUMU$_\textit{PES}$ still achieves the fastest runtime due to the most reasonable pruning strategies it uses. Moreover, on each dataset, as \textit{minuo} decreases, it does not generate many more candidate patterns. Upper bounds \textit{PES} and \textit{RSS} already make it possible to reduce many invalid candidate patterns.
	
	\item Although \textit{TPUO} and \textit{TSUO} are tighter upper bounds, as \textit{minuo} decreases, they still do not reduce many irrelevant candidate patterns for SUMU$_\textit{TPUO}$. 	
\end{enumerate}

\subsection{Memory Evaluation}

The memory consumption of each variant of SUMU is close and fluctuating, and we present the approximate memory consumption under different datasets. We investigate the reasons for the disparities in memory consumption based on program design details. The experimental results regarding memory consumption are shown in Table \ref{memory}.

\begin{table}[h]
	\centering
	\caption{Memory consumption}
	\label{memory}
	\resizebox{\columnwidth}{!}{
		\begin{tabular}{c|cccccc}
			\hline \hline
			\multirow{2}{*}{}   & \multicolumn{6}{c}{\textbf{Approximate memory consumed (MB)}} \\ \cline{2-7} 
			&          Bible         &  FIFA & Sign  & Syn10k  & Syn20k  & Syn40k  \\ \hline
			{fixed \textit{minuo}}  & 1,000 $\sim$ 1,400 & 1,400 $\sim$ 1,700 & 200 $\sim$ 400 & 300 $\sim$ 800 & 600 $\sim$ 800 & 1,200 $\sim$ 1,500  \\ \hline
			{fixed \textit{minsup}}  & 1,000 $\sim$ 1,400 & 1,500 $\sim$ 1,700 & 200 $\sim$ 400 & 300 $\sim$ 600 & 300 $\sim$ 800 & 1,200 $\sim$ 1,500    \\ \hline
			\hline 
		\end{tabular}
	}
\end{table}

Since each variant of SUMU uses both UOL-Chain and UO-Table, they consume little difference in memory, and the difference is within a reasonable range. The variants of SUMU employ a different number of pruning strategies, and thus they differ somewhat in the use of some auxiliary data structures. If the pruning Strategy \ref{strategy1} is used, then unpromising items should be filtered. To find out which items are unpromising, the program utilizes a hash table to record the support for each item. On the contrary, if all items are used directly, it is sufficient for the program to use a single list to record those items that occur in the sequence database. The difference between SUMU$_{\textit{PEUO}}$ and SUMU$_{\textit{TPUO}}$ is that they use different upper bounds and pruning strategies. For the calculation of \textit{PEUO} and \textit{RSUO} of a pattern, this is relatively simple. Scanning the UOL-Chain of a pattern quickly and accumulating corresponding values. However, for the calculation of \textit{TPUO} and \textit{TSUO} of a pattern, several \textit{minsup}-sized priority queues are required. This allows computing tighter upper bound values, but also consumes additional memory space. As for SUMU$_{\textit{PES}}$, it uses more upper bounds \textit{PES} and \textit{RSS} (adopts more pruning Strategies \ref{strategy6} and \ref{strategy7}) compared to SUMU$_{\textit{PEUO}}$. This suggests that in pattern extension, the program needs the associated hash tables to decide which candidate patterns satisfy the upper bounds \textit{PES} and \textit{RSS}. It seems that the more upper bounds and pruning strategies are used, the more memory is consumed. Nevertheless, in the experiments, effective pruning strategies can avoid unnecessary UOL-Chain and UO-Table builds due to the non-generation of some candidate patterns, also saving memory consumption. Therefore, the memory consumption of each variant of SUMU is roughly equal.

\subsection{Scalability}

This section selected five synthetic datasets to evaluate the scalability of each variant of SUMU. The dataset size increases from 10k to 50k sequence records, increasing by 10k each time. We set a relative support for experiments, e.g., \textit{minsup} was set to 10, 20, 30, 40, and 50 for the five synthetic datasets, respectively. In addition, the \textit{minuo} is set to 0.1 in order to generate more HUOSPs. We analyze the scalability in terms of runtime and candidate pattern generation, and the experimental results are shown in Fig. \ref{scalability}.

\begin{figure}[h]
	\centering
	\includegraphics[trim=0 0 0 0,clip,scale=0.4]{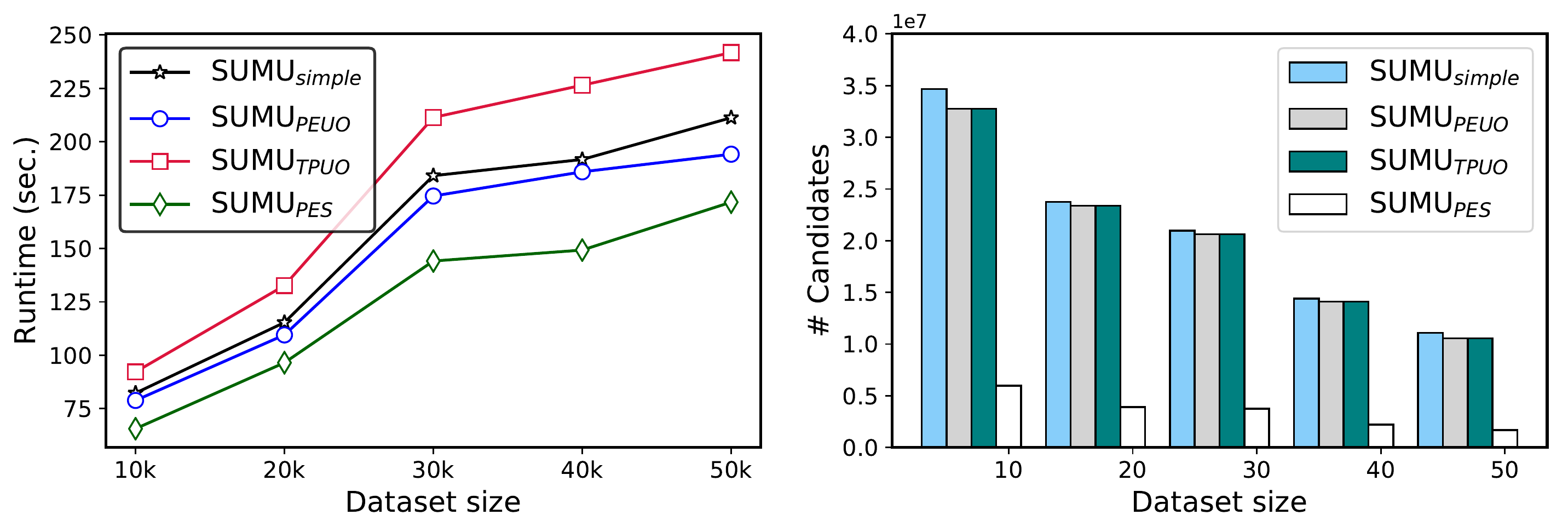}
	\caption{Scalability of the compared variants of SUMU}
	\label{scalability}
\end{figure}

From Fig. \ref{scalability}, it is clear that the runtime of each variant of SUMU grows as the size of the processed dataset increases. This is consistent with our assumption that larger datasets carry more candidate patterns, increasing the processing difficulty. The use of UOL-Chain and UO-Table makes the trend of each variant of SUMU the same, with only differences in efficiency. The difference between all SUMU variants is clear, with SUMU$_{\textit{PES}}$ performing best and SUMU$_{\textit{TPUO}}$ performing worst. For SUMU$_{\textit{PES}}$, there is no such rapid growth of candidate patterns. While other variants of SUMU generate a large number of candidate patterns. Therefore, it performs well when handling large-scale datasets. The large number of sorting operations required for the calculation of tighter upper bounds, in particular, causes SUMU$_{\textit{TPUO}}$ to perform poorly. The difference between SUMU$_{\textit{PEUO}}$ and SUMU$_{\textit{simple}}$ illustrates the effectiveness of the pruning strategy \ref{strategy1}.

\section{Conclusions and Future Work}
\label{sec:conclusion}

In this paper, to address the problem of sequence utility maximization, we formulate the problem of HUOSPM using the utility occupancy measure. Our definitions allow SPM to take into account the utility share of a pattern in sequence records and the database, thus allowing derived sequential patterns to carry more useful information. Furthermore, we proposed a novel algorithm called SUMU. In the mining process, SUMU employs two compact data structures, UOL-Chain and UO-Table, as well as associated pruning strategies. We consider the possible support and utility occupancy of a candidate sequence, and thus six upper bounds are designed. Extensive experiments on real and synthetic datasets demonstrate that SUMU can efficiently discover all interesting HUOSPs and has better scalability. Utility occupancy is a useful measure to process the sequence data. We can do more interesting research in the future. Explorations of HUOSPM can be developed in a distributed environment or under privacy protection. Moreover, some issues, such as the rare item problem and the neglect of recency, are also interesting to be studied in HUOSPM.

\section*{Acknowledgment}

This research was supported in part by the National Natural Science Foundation of China (Nos. 62272196 and 62002136), Natural Science Foundation of Guangdong Province (No. 2022A1515011861), Guangzhou Basic and Applied Basic Research Foundation (No. 202102020277), and the Young Scholar Program of Pazhou Lab (No. PZL2021KF0023).

%%
%% The next two lines define the bibliography style to be used, and
%% the bibliography file.
\bibliographystyle{ACM-Reference-Format}
\bibliography{SUMU}

%%
%% If your work has an appendix, this is the place to put it.
% \appendix

\end{document}